%% file: Main_ArXiv.tex
\documentclass[aps,prl,reprint,superscriptaddress,twocolumn,notitlepage]{revtex4-2}

\usepackage{tikz}  
\usetikzlibrary{arrows,shapes,positioning,shadows,backgrounds,fit}

\usepackage{float}
\usepackage{relsize}
\usepackage{graphicx}
\usepackage{amsmath, bbm}
\usepackage{amssymb}
\usepackage{amsthm}
\usepackage{comment}
\usepackage{mathrsfs}
\usepackage{xcolor}
\usepackage{cancel}
\usepackage{bbold}
\usepackage{titlesec}
\usepackage{enumitem}  
\usepackage{stmaryrd}

\usepackage[normalem]{ulem}

\titlespacing{\section}{0ex}{2ex}{0.4ex}
\titleformat{name=\section}{\bf}{\thesection.}{.3em}{}

\usepackage{soul}
\usepackage{dsfont}
\usepackage{times}

\def\be{\begin{eqnarray}}
\def\ee{\end{eqnarray}}
\newcommand{\Tr}[1]{\mathrm{Tr}\left[#1\right]}

\newcommand{\ket}[1]{|{#1}\rangle}
\newcommand{\bra}[1]{\langle{#1}|}

\newcommand{\D}{\mathsf{D}}
\newcommand{\F}{\mathsf{F}}
\newcommand{\Werg}{W_{\mathrm{erg}}}

\newcommand{\Emin}{E_{\mathrm{min}}}

\newcommand{\est}{M,\Pi}                  
\newcommand{\rhoI}{\rho_{\est}}           
\newcommand{\UWI}{U_{W}^{(\est)}}         
\newcommand{\UWIdag}{U_{W}^{(\est)\dagger}}
\newcommand{\WI}{W_{\est}}                
\newcommand{\WItot}{W_{\est}^{\mathrm{tot}}} 
\newcommand{\etaI}{\eta_{\est}}           
\newcommand{\sigmaI}{\sigma_{\est}}       
\newcommand{\rrec}[1]{r^{(\est)}_{#1}}    
\newcommand{\Prec}[1]{P^{(\est)}_{#1}}

\theoremstyle{plain}

\newtheorem{lem}{Lemma}

\newtheorem{thm}{Theorem}

\definecolor{myblue}{rgb}{0.2,0.2,0.8}
\definecolor{myblack}{rgb}{0,0,0}
\definecolor{myurl}{rgb}{0.1,0.1,0.4}

\usepackage[colorlinks=true,citecolor=myblue,linkcolor=myblack,urlcolor=myurl]{hyperref}

\usepackage{bm}
\makeatletter
\newenvironment{supplementbib}[1]
  {\par\vskip1.0em\noindent
   \list{\@biblabel{\@arabic\c@enumiv}}%
        {\settowidth\labelwidth{\@biblabel{#1}}%
         \leftmargin\labelwidth\advance\leftmargin\labelsep
         \usecounter{enumiv}%
         \let\p@enumiv\@empty
         \renewcommand\theenumiv{\@arabic\c@enumiv}}%
   \footnotesize
   \sloppy\clubpenalty4000\@clubpenalty\clubpenalty\widowpenalty4000%
   \sfcode`\.\@m}
  {\def\@noitemerr{\@latex@warning{Empty `supplementbib' environment}}\endlist}
\makeatother

\begin{document}

\title{Quantum work extraction from partial information and with finite resources}

\author{Giacomo Guarnieri}

\affiliation{Department of Physics A. Volta, University of Pavia, Via Bassi 6, 27100, Pavia, Italy}
\affiliation{INFN Sezione di Pavia, Via Agostino Bassi 6, I-27100, Pavia, Italy}

\author{Diego Maragnano}

\affiliation{Department of Physics A. Volta, University of Pavia, Via Bassi 6, 27100, Pavia, Italy}

\author{Giulia Gamba}

\affiliation{Department of Physics A. Volta, University of Pavia, Via Bassi 6, 27100, Pavia, Italy}

\author{Lorenzo Zoppelletto}

\affiliation{Department of Physics A. Volta, University of Pavia, Via Bassi 6, 27100, Pavia, Italy}

\author{Marco Liscidini}

\affiliation{Department of Physics A. Volta, University of Pavia, Via Bassi 6, 27100, Pavia, Italy}

\begin{abstract}
Information can be converted into work, but in quantum mechanics information about a state is not freely available: it must be inferred statistically from measurements on a finite number of copies. We study work extraction in this finite-resource setting by introducing a partial-information and finite-resources (PIFR) quantum Maxwell’s demon. Given $N$ identical copies of a state with known Hamiltonian, the demon measures $M$ copies to estimate the state and the corresponding ergotropic unitary, which is then applied to the remaining $N-M$ copies. This protocol induces a trade-off between information acquisition, reconstruction accuracy, and thermodynamic yield, making the total extracted work normalized to the ideal ergotropic benchmark the relevant figure of merit. As our central result, we derive a universal closed-form trade-off bound that places this ergotropic efficiency between a Carnot-type ceiling $1-M/N$ and a floor controlled by a reconstruction precision rooted in finite-sample quantum estimation theory; optimizing the copy allocation yields $M^{\ast}\propto N^{2/3}$ and an $N^{-1/3}$ approach to the ideal limit, set by a conservative, worst-case reconstruction precision. By considering standard quantum state tomography, we numerically verify the presence of an optimal resource distribution, which also depends on the purity of the state under consideration. Our results identify finite-copy work extraction as a genuinely task-dependent inference problem, in which estimation strategies should be judged by thermodynamic performance rather than reconstruction fidelity alone.
\end{abstract}

\maketitle

\textit{Introduction. --}
Since its original inception, Maxwell’s demon has provided the paradigmatic setting in which information and thermodynamics confront each other: a hypothetical agent that acquires microscopic information about a system can in principle exploit it to extract work, seemingly violating the Second Law of Thermodynamics~\cite{maxwell2001theory,maruyama2009colloquium}. This apparent paradox was resolved by the works of Szilard~\cite{szilard1929entropieverminderung,szilard1964decrease}, Landauer~\cite{landauer1961irreversibility} and Bennett~\cite{bennett1973logical,bennett1982thermodynamics}, who revealed that the extracted energy depends on the information collected by the demon and that erasing this information carries an irreversible energetic cost~\cite{delrio2011thermodynamic,faist2015minimal} that balances or outweighs the extractable work~\cite{parrondo2015thermodynamics,sagawa2008second,toyabe2010experimental,koski2014experimental,berut2012experimental}.

Unlike in the classical regime in which Maxwell's demon was originally proposed, quantum mechanics dictates that acquiring information about a system is itself not free~\cite{goold2016role,vinjanampathy2016quantum,horodecki2013fundamental,brandao2015second}. Complete knowledge of a quantum state can only be obtained statistically from measurements performed on many identically prepared copies of the system. Moreover, quantum measurements generally modify the state itself. After a projective measurement, for instance, a copy is no longer left in the original state, but collapses onto the eigenstate associated with the observed outcome. The copies used to acquire information are therefore effectively consumed by the measurement process and can no longer be used for work extraction.

As a consequence, when only a finite number $N$ of copies is available, some of them must necessarily be sacrificed to infer the information required to determine the work-extraction protocol, while only the remaining copies can actually be exploited to extract work. This raises a fundamentally new thermodynamic question that is absent in the classical formulation of Maxwell's demon, where only the storage and manipulation of information carry a cost: what is the trade-off between acquiring information and exploiting it for work extraction? Beyond its fundamental relevance, this question is also important from a practical perspective. In realistic quantum technologies resources are often scarce, and one must determine how they should be optimally divided between information acquisition and work extraction.

\begin{figure}[t!]
    \centering
    \includegraphics[width=\linewidth]{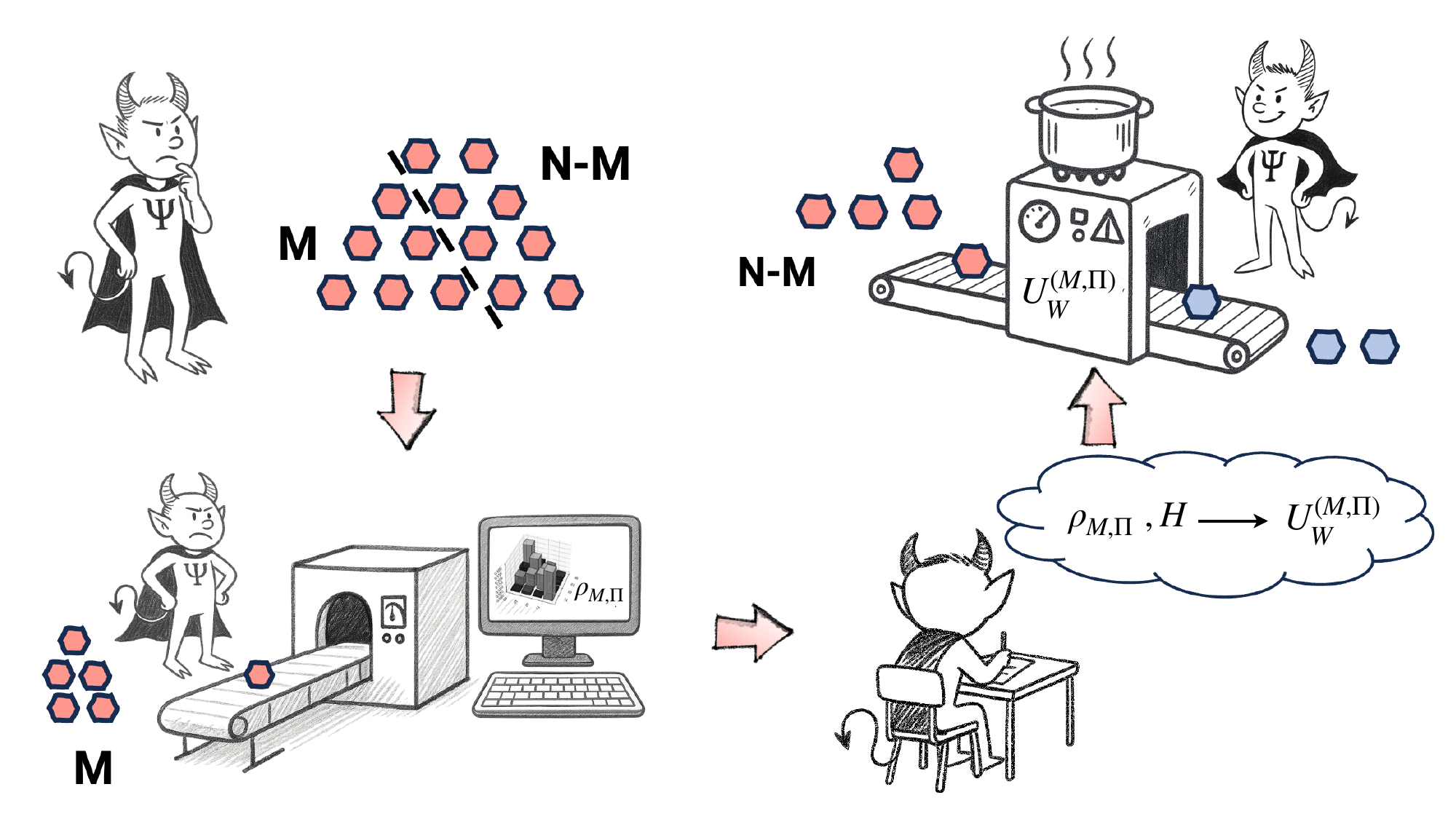}
    \caption{Schematic representation of a Maxwell's demon that is constrained to work with only a finite number $N$ of copies of a quantum system. The demon will have to perform some measurements on $M<N$ in order to reconstruct an estimate of the state and then use this acquired \textit{partial} information in order to unitarily extract the maximum work possible (ergotropy) from the remaining $N-M$ copies.}
    \label{fig:figure1}
\end{figure}

To address this question, we consider the following thermodynamic task. An agent is given $N$ identical copies of an unknown quantum state $\rho$, governed by a known Hamiltonian $H$ (see Fig.~\ref{fig:figure1}). The goal is to extract as much work as possible through unitary operations. In the ideal case in which $\rho$ is perfectly known, the maximum extractable work is given by the ergotropy of the state~\cite{pusz1978passive,lenard1978thermodynamical,skrzypczyk2014work}, and it can be achieved through the corresponding optimal unitary transformation. In practice, however, determining this optimal operation requires information about the state itself ~\cite{allahverdyan2004maximal}.

We therefore introduce a partial-information and finite-resources (PIFR) quantum Maxwell's demon. Unlike an idealized omniscient demon, the PIFR demon can only access partial information about $\rho$, obtained from measurements performed on a subset $M$ of the available copies. From this incomplete information, it reconstructs an estimate of the state, infers the corresponding work-extraction unitary, and applies it to the remaining $N-M$ untouched copies.

The performance of the PIFR demon is governed by a fundamental finite-resource trade-off. Measuring more copies generally improves the inferred control operation, but leaves fewer copies available for work extraction. Measuring fewer copies has the opposite effect: more copies remain available for work extraction, but the inferred operation becomes less accurate. The central question is therefore how finite resources should be optimally divided between learning and thermodynamic exploitation.

This finite-resource regime also changes the meaning of optimal state reconstruction itself. In standard quantum state tomography~\cite{james2001qst,gross2010quantum}, the quality of a protocol is naturally quantified by the fidelity between the reconstructed and true states. Here, however, information acquisition is not an end in itself, but a resource for a subsequent thermodynamic task. When only partial information can be acquired, reconstruction fidelity alone is therefore no longer the most relevant figure of merit. What ultimately matters is how effectively the acquired information can be converted into extractable work~\cite{huang2020predicting,aaronson2020shadow,elben2023randomized}.

Conceptually related investigations have recently appeared: finite-sample bounds for quantum metrology were established in Ref.~\cite{meyer2025quantum}; the certification of ergotropy and of many-body properties under informationally incomplete or imperfect measurements was developed in Refs.~\cite{pagliaro2026certifying,zambrano2026certification,mortimer2026bounding}; and the sample complexity of black-box work extraction was analyzed in Ref.~\cite{chakraborty2025sample}. Our work complements these results: rather than certifying a lower bound on ergotropy from fixed data, we quantify the unavoidable inefficiency of the PIFR protocol as an information-thermodynamic tradeoff. We first derive a universal closed-form bound (Theorem~\ref{thm:main} below) relating the efficiency of the protocol to $N$, $M$, and a reconstruction precision $\varepsilon_M$ certified by the chosen estimation strategy. We then make the task dependence concrete by considering a demon reconstructing the state with standard quantum state tomography (QST). 

\textit{PIFR Quantum Maxwell's demon and ergotropy.}---Consider a quantum
system in a $D$-dimensional Hilbert space and with Hamiltonian
\begin{equation}
H=\sum_{k=1}^{D}\epsilon_k\,|\epsilon_k\rangle\!\langle\epsilon_k|,
\qquad
\epsilon_1\leq \epsilon_2\leq \cdots \leq \epsilon_D,
\label{eq:H}
\end{equation}
which we assume to be known. The state, instead, is unknown
and denoted by $\rho$. The demon has access to $N$ i.i.d.\ non-interacting copies
of the system, $\rho^{\otimes N}$. A measurement performed on one
copy is assumed to be destructive, in the sense that the measured
copy cannot subsequently be used for work extraction. We further assume that work extraction is performed locally on the remaining copies, by applying the same single-copy unitary to each unmeasured copy.

For a given state $\rho$ and Hamiltonian $H$, the natural ideal
benchmark is the ergotropy, i.e.\ the maximum work extractable from
$\rho$ by unitary operations~\cite{allahverdyan2004maximal},
\begin{equation}
W_{\rm erg}(\rho,H)
=
\max_{U}\,
\Tr{(H-U^{\dagger}HU)\rho}
=
\Tr{\rho H}-E_{\min}(\rho),
\label{eq:ergotropy}
\end{equation}
with
\begin{equation}
E_{\min}(\rho):=\min_{U}\Tr{U\rho U^{\dagger}H}.
\label{eq:Emin}
\end{equation}
If $\rho=\sum_{k}r_k|r_k\rangle\!\langle r_k|$, with
$r_1\ge r_2\ge\cdots\ge r_D$, the corresponding ergotropic unitary
can be written as
\begin{equation}
U_W(\rho)=\sum_{k=1}^{D} |\epsilon_k\rangle\!\langle r_k|,
\label{eq:UW}
\end{equation}
which rearranges the populations of $\rho$ into the passive state
associated with $H$. 

An oracle-demon with exact prior knowledge of $\rho$ would not need to
consume any copies in order to infer the work-extraction protocol. It could
instead compute the ergotropic unitary $U_W(\rho)$ directly and apply it to
all $N$ available copies. This gives the ideal upper bound on the total
extractable work,
\begin{equation}
W_{\rm bound}(N)=N\,W_{\rm erg}(\rho,H).
\label{eq:Wideal}
\end{equation}

The PIFR demon, however, operates under stricter constraints. Out of the
$N$ available copies, it uses $M<N$ of them to acquire partial information
about the state through a measurement scheme $\Pi$. The resulting
reconstruction is denoted by $\rhoI $, emphasizing its dependence on
both the number of measured copies and the chosen measurement scheme. The
demon then computes the ergotropic unitary associated with this reconstructed
state,
\begin{equation}
\UWI :=U_W(\rhoI ) .
\label{eq:UWMPi}
\end{equation}

This unitary is then applied to the remaining $N-M$ untouched copies
of the true state $\rho$. The extracted work per copy is thus
\begin{equation}
\WI 
=
\Tr{\rho H}
-
\Tr{\left(
\UWI 
\rho
\UWIdag 
\right)H},
\label{eq:WMPi}
\end{equation}
and, within the present local-extraction protocol, the total extracted
work is
\begin{equation}
\WItot (N)=(N-M)\,\WI .
\label{eq:WMPitot}
\end{equation}

Since $\UWI $ is optimal for $\rhoI $ rather than for the true
state $\rho$, one always has
\begin{equation}
\WI \leq W_{\rm erg}(\rho,H),
\label{eq:upperWMPi}
\end{equation}
with equality in the ideal case $\rhoI =\rho$.

The central figure of merit of this work is the \emph{ergotropic
efficiency}, defined as the extracted work normalized to the ideal upper
bound,
\begin{equation}
\etaI :=
\frac{\WItot (N)}{W_{\rm bound}(N)}
=
\left(1-\frac{M}{N}\right)
\frac{\WI }{W_{\rm erg}(\rho,H)}.
\label{eq:efficiency}
\end{equation}

This quantity directly captures the competition between learning and
exploitation. If $M=N$, all copies are consumed by the reconstruction
stage and no copy is left for work extraction, so that
$\etaI =0$. At the opposite extreme, unit efficiency would require
perfect state knowledge without consuming any copies, corresponding to
the oracle limit $M=0$ and $\WI =W_{\rm erg}(\rho,H)$. In any
realistic finite-copy setting, however, increasing $M$ improves the
reconstruction, and hence the inferred unitary, at the price of leaving
fewer copies from which to extract work; conversely, small $M$ leaves
more copies for extraction but typically yields a less accurate
work-extraction unitary.

After consuming $M$ copies, even if the reconstruction allows the demon to determine and implement a unitary that coincides with the true ergotropic unitary of $\rho$\footnote{This may occur in exceptional cases even when $\rhoI \neq\rho$, e.g.\ when the two states share their eigenbasis and eigenvalue ordering.}, the efficiency is still bounded
by
\begin{equation}
\etaI \leq 1-\frac{M}{N}.
\label{eq:copy_bound}
\end{equation}
One of the main results of this Letter is to quantify \emph{how much
lower} than this copy-counting bound the efficiency is driven by the
finite accuracy of the reconstruction.

\textit{Finite-sample tradeoff bound. --}
We now derive our central theoretical result, which answers this question in a closed form. The key step is to bound the \emph{ergotropic deficit} $\Delta W:=\Werg(\rho,H)-\WI \ge 0$ in terms of a distance between the reconstructed and true states.

\begin{lem}[Ergotropic-deficit bound]\label{lem:deficit}
For any reconstructed state $\rhoI $,
\begin{equation}
\Delta W\;\le\; 2\,\Delta_{H}\,\D(\rho,\rhoI ),
\label{eq:deficit-bound}
\end{equation}
where $\Delta_{H}:=\epsilon_{D}-\epsilon_{1}$ is the spectral spread of $H$ and $\D(\rho,\sigma):=\tfrac12\|\rho-\sigma\|_{1}$ is the trace distance.
\end{lem}

\noindent\emph{Proof sketch.} Writing $\Delta W=\{\Tr{H\sigmaI }-\Emin(\rhoI )\}+\{\Emin(\rhoI )-\Emin(\rho)\}$ with $\sigmaI :=\UWI \rho\, \UWIdag $, each term is bounded by $\Delta_{H}\D(\rho,\rhoI )$ via H\"older's inequality applied to the centred Hamiltonian $H-\tfrac12(\epsilon_{1}+\epsilon_{D})\mathbb{1}$, of spectral radius $\Delta_{H}/2$, together with the Lipschitz continuity of $\Emin$; the full proof is given in the Supplemental Material (SM)~\cite{SM}. $\hfill\square$

The strength of Lemma~\ref{lem:deficit}, which can be regarded as the ergotropic counterpart of the Fannes--Audenaert inequality for entropies, lies in its universality: it applies to any reconstruction strategy $\mathcal{R}$ that achieves trace distance $\varepsilon_M$ with probability $\ge 1-p_{\mathrm{fail}}$ on $M$ copies. Combining it with Eqs.~\eqref{eq:efficiency} and \eqref{eq:upperWMPi}, we obtain our main result.

\begin{thm}[Finite-sample tradeoff bound]\label{thm:main}
Define the dimensionless \emph{ergotropic susceptibility}
\begin{equation}
\Lambda(\rho,H):=\frac{2\Delta_{H}}{\Werg(\rho,H)}.
\label{eq:Lambda}
\end{equation}
Under the PIFR protocol, with probability at least $1-p_{\mathrm{fail}}$ over the randomness of the measurement outcomes,
\begin{equation}
\left(1-\frac{M}{N}\right)\!\bigl[1-\Lambda(\rho,H)\,\varepsilon_{M}\bigr]_{+}
\!\le\,\etaI \,\le\,1-\frac{M}{N},
\label{eq:main-bound}
\end{equation}
where $[x]_{+}:=\max(x,0)$, and $\varepsilon_{M}\equiv\varepsilon_{M}(\mathcal{R},p_{\mathrm{fail}})$ is a high-probability upper bound on $\D(\rho,\rhoI )$ certified by the reconstruction strategy $\mathcal{R}$ with $M$ copies.
\end{thm}

\noindent The upper bound in Eq.~\eqref{eq:main-bound} is the copy-counting ceiling of Eq.~\eqref{eq:copy_bound}, saturated in the oracle limit $\varepsilon_M\to 0$; the lower envelope quantifies the \emph{unavoidable inefficiency} imposed by the finiteness of $M$. The name given to $\Lambda$ follows linear-response theory: by Lemma~\ref{lem:deficit}, $\Lambda$ is the coefficient that converts a small reconstruction error into the leading-order relative work loss, exactly as a magnetic susceptibility converts a weak applied field into a magnetization change. Nearly passive states ($\Werg\to0$ at fixed $\Delta_H$) are maximally susceptible; equivalently, $1/\Lambda$ is the largest reconstruction error compatible with a nontrivial guarantee.
\begin{figure*}[!ht]
    \centering
    \includegraphics[width=0.99\linewidth]{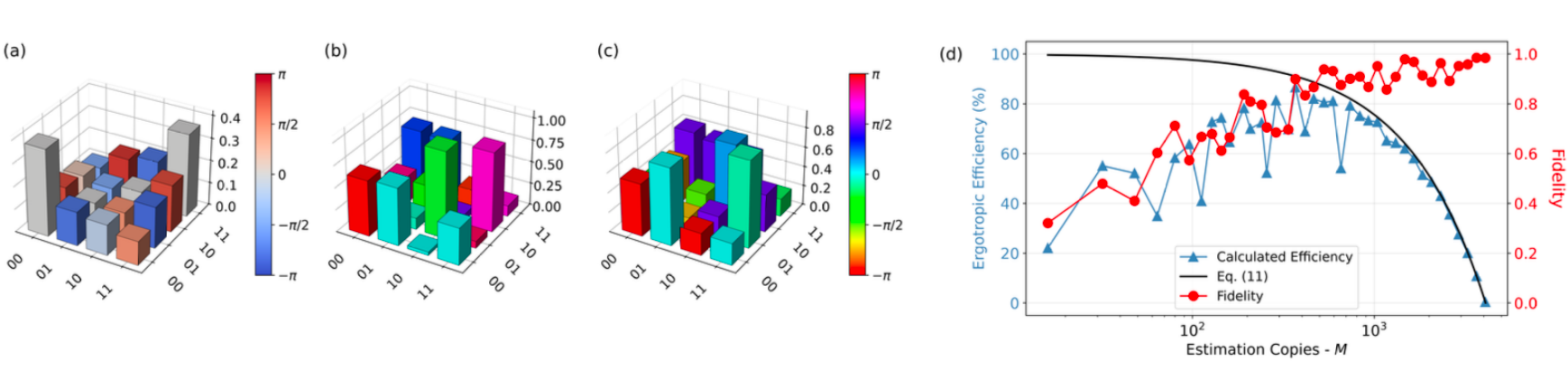}
    \caption{Steps of the PIFR protocol with QST for a single two-qubit density matrix. (a) Density matrix reconstructed with QST using 5 copies per projective measurement (i.e.\ $M=80$). (b) Ideal ergotropic unitary matrix. (c) Ergotropic unitary matrix computed from the reconstructed state of panel (a). (d) Efficiency (left vertical axis) and fidelity between the reconstructed and the true state (right vertical axis) as a function of the number $M$ of copies used to reconstruct the state, for $N=4096$. The black solid line is the copy-counting bound $1-M/N$ of Eq.~\eqref{eq:copy_bound}.}
    \label{fig:figure2}
\end{figure*}

Theorem~\ref{thm:main} becomes fully explicit once $\varepsilon_M$ is computed. The fundamental, strategy-independent benchmark follows from the finite-sample estimation bounds of Ref.~\cite{meyer2025quantum}: treating state reconstruction as the simultaneous estimation of the $D^2-1$ generalized Bloch parameters~\cite{paris2009quantum,giovannetti2011advances}, converting fidelity into trace distance through the Fuchs--van de Graaf inequality~\cite{fuchs1999cryptographic}, and allocating the confidence budget optimally across parameters, one obtains, to leading order in $1/M$, $\varepsilon_{M}=\tfrac12\left[(D^{2}-1)\log[(D^{2}-1)/(4p_{\mathrm{fail}})]/M\right]^{1/2}$. Remarkably, the eigenvalues of the quantum Fisher information matrix cancel exactly in this expression; moreover, its $\left[(D^{2}-1)/M\right]^{1/2}$ scaling is robust against the measurement-incompatibility subtleties of multi-parameter quantum estimation~\cite{szczykulska2016multi}, matching the asymptotic Holevo Cram\'er--Rao rate up to a logarithmic factor~\cite{SM}. Protocol-specific certificates of the same $M^{-1/2}$ form, with larger dimensional prefactors, follow from Hoeffding-type concentration for the tomographic schemes used below~\cite{pagliaro2026certifying,zambrano2026certification,mortimer2026bounding,SM}.

Two structural consequences follow. \emph{(i) Optimal allocation.} Substituting $\varepsilon_{M}\simeq A/\sqrt{M}$---with $A=\tfrac{\Lambda}{2}\left[(D^{2}-1)\log[(D^{2}-1)/(4p_{\mathrm{fail}})]\right]^{1/2}$ for the benchmark above---and maximizing the lower envelope over $M$ yields
\begin{equation}
M^{\ast}\simeq \left(\frac{AN}{2}\right)^{2/3},\qquad \etaI^{\max}\simeq 1-\tfrac{3}{2^{2/3}}A^{2/3}N^{-1/3}.
\label{eq:Mstar}
\end{equation}
The approach to the oracle limit is therefore $1-\mathcal{O}(N^{-1/3})$: slower than both the standard-quantum-limit learning rate $M^{-1/2}$ and the linear exploitation cost $1-M/N$, because the optimum must balance the two. This is the thermodynamic price of inference. \emph{(ii) Regime of validity.} The lower envelope is non-trivial only when $\Lambda(\rho,H)\varepsilon_{M}<1$, i.e.\ when $M>(\Delta_{H}/\Werg)^{2}(D^{2}-1)\log[(D^{2}-1)/(4p_{\mathrm{fail}})]$. For near-passive states $(\Werg\to 0)$ this threshold diverges, consistent with the black-box work-extraction impossibility of Ref.~\cite{chakraborty2025sample}.

Theorem~\ref{thm:main} shifts the operational meaning of ``good reconstruction'' in the PIFR context: minimizing $\varepsilon_{M}$ in \emph{trace distance} (or equivalently maximizing fidelity) is, to leading order, the correct objective for thermodynamic task performance---but only when weighted by the state-dependent factor $\Lambda(\rho,H)$. A reconstruction strategy that is globally fidelity-optimal can be thermodynamically suboptimal if it wastes measurement effort on features of $\rho$ that carry small weight in $\Lambda$; conversely, a state-dependent, task-adapted strategy could close most of the gap to the copy-counting bound even with very few copies \cite{binosi2024tailor,caruccio2025experimental}
\\

We now illustrate a practical implementation of the PIFR demon on a system of $n=2$ qubits, with Hilbert-space dimension $D=2^{n}=4$. To implement the protocol, the demon reconstructs the density matrix using conventional QST and determines the corresponding ergotropic unitary from Eq.~\eqref{eq:UWMPi}.

\textit{Implementation of the PIFR demon. --} We consider a Hamiltonian diagonal in the computational basis,
\begin{align}
    H = \sum_{k=0}^{D-1} \epsilon_k \, |k\rangle\langle k| ,
    \label{eq:RandomH}
\end{align}
with energies $\{\epsilon_k\}_{k=0}^{D-1}$ drawn independently from a standard normal distribution, $\epsilon_k \sim \mathcal{N}(0,1)$, and sorted in ascending order. These values were chosen arbitrarily, subject only to the requirement that the spectrum be non-degenerate. In our calculations, the energies were fixed once for all numerical experiments to be $\{\epsilon_k\} = \{-0.1321,\ 0.1049,\ 0.1257,\ 0.6404\}$,
which fully specifies $H$ for reproducibility purposes. 

For each input state $\rho$, the demon is provided with $N$ i.i.d.\ copies of the system. A subset of $M<N$ copies is consumed to reconstruct $\rho$, whereas the remaining $N-M$ copies are retained for work extraction. In the single-state illustration of Fig.~\ref{fig:figure2} and in the QST benchmark of Fig.~\ref{fig:figure3}(a,b), we fix $N=4096$. We vary $M$ to explicitly illustrate the trade-off between reconstruction accuracy and the number of copies available for work extraction.

Because the performance of the PIFR demon is state dependent, we assess its typical behavior by averaging over ensembles of $2{,}000$ density matrices, which we verified to be statistically representative. We consider two ensembles: (i) pure states and (ii) mixed states whose purities, $\Tr{\rho^2}$, are uniformly distributed over the allowed interval $[1/D,1]$. 

Since the probe states $\rho$ are themselves Haar-random, the relative orientation between the eigenbases of $\rho$ and $H$ is randomized by the state ensemble, and once the ergotropic efficiency is averaged over a sufficiently large number of such states the resulting curves become essentially indistinguishable from those obtained with any other choice of non-degenerate spectrum for $H$. 
Degeneracies in the spectrum of $H$, by contrast, should be treated with care, since they leave the passive-state ordering non-unique within the degenerate subspace, allowing reconstruction noise to bias the estimated efficiency in a way that does not necessarily average away.

For two qubits, conventional QST involves $16$ independent projective measurements. The $M$ copies used for state reconstruction are distributed uniformly among these measurements, and Poissonian statistics are assumed for the corresponding measurement counts. Consequently, throughout our calculations, $M$ is restricted to integer multiples of $16$, with the smallest value, $M=16$, corresponding to one copy per projective measurement.

Figure~\ref{fig:figure2}(a) shows the density matrix of a randomly selected state reconstructed using $M=80$ copies, corresponding to five copies per projective measurement. The ergotropic unitary associated with the exact state is shown in panel (b), whereas the unitary inferred from the reconstructed state is shown in panel (c). In Fig.~\ref{fig:figure2}(d), we report the Uhlmann fidelity between the true and reconstructed states, together with the corresponding ergotropic efficiency, as functions of $M$. As expected, the fidelity improves overall as more copies are used for the reconstruction. The point-to-point fluctuations observed in both quantities arise from the Poissonian noise affecting the simulated measurement counts.
\begin{figure*}
    \centering
    \includegraphics[width=0.99\linewidth]{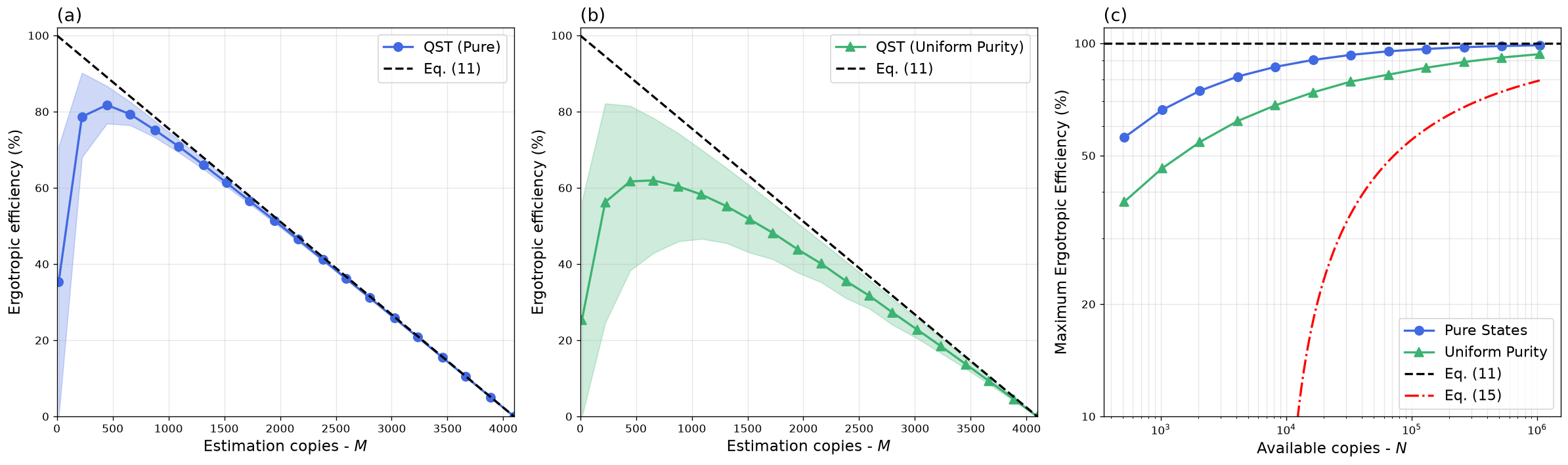}
    \caption{QST benchmarking of the PIFR demon against the theoretical bounds. (a) Average ergotropic efficiency for two-qubit pure states vs. the number $M$ of copies used for the reconstruction, with $N=4096$ available copies, averaged over $2{,}000$ random pure states; the shaded band indicates one standard deviation. (b) Same as (a) for $2{,}000$ two-qubit states drawn from a distribution uniform in purity. In both panels the dashed black line is the copy-counting bound $1-M/N$ of Eq.~(11). (c) Maximum ergotropic efficiency, optimized over $M$, as a function of the number $N$ of available copies, up to $N=1{,}048{,}576$, for the pure (blue) and uniform-purity (green) ensembles. The dashed black line is the oracle limit, Eq.~(11). The dash-dotted red curve is the lower bound predicted by Eq.~(15), which is valid only for large N. In (15) we use $A$ computed from first principles for the uniform-purity ensemble ($\Delta H=0.77$, $\langle 1/W_{\rm erg}\rangle=5.87$, $p_{\rm fail}=0.05$), giving $A\simeq36.5$.}
    \label{fig:figure3}
\end{figure*}

For $M=1024$, corresponding to $64$ copies per projective measurement, the fidelity is already of the order of $0.9$ and continues to approach unity as $M$ increases, reaching values above $0.95$ for some of the largest values of $M$ considered. Remarkably, for the state shown in Fig.~\ref{fig:figure2}, the ergotropic efficiency reaches its maximum, approximately $80\%$, using only a few hundred copies, when the reconstruction fidelity is only about $60\%$. Thus, a nearly faithful reconstruction is not required to achieve optimal thermodynamic performance. Beyond an intermediate value of $M$, the improvement in the inferred unitary no longer compensates for the reduction in the number of copies available for work extraction.

Having examined a particular state, we now turn to general conclusions that do not depend on the specific state.
Fig.~\ref{fig:figure3} shows the average ergotropic efficiency over ensembles of $2{,}000$ two-qubit states. For $N=4096$, both the pure-state ensemble (a) and the ensemble uniform in purity (b) exhibit the non-monotonic dependence on $M$ as seen for the particular case of Fig. \ref{fig:figure2}. However, as a consequence of the averaging, now the curve is smooth. When $M$ is small, the limited measurement statistics yield an inaccurate work-extraction unitary; increasing $M$ initially improves the efficiency, which reaches approximately $80\%$ for pure states and $60\%$ for the uniform-purity ensemble after using only a few hundred copies. Beyond this optimum, the improvement in the inferred unitary is outweighed by the decreasing fraction $1-M/N$ of copies available for work extraction, and the efficiency eventually vanishes at $M=N$. Pure states remain closer to the oracle bound and display a considerably narrower distribution, indicating that the thermodynamically relevant unitary can be inferred more reliably for this ensemble, whereas generic mixed states exhibit a stronger state-to-state dependence. 

In Fig.~\ref{fig:figure3}(c) we show the maximum efficiency as a function of the total number $N$ of available copies for both ensembles. As expected, the efficiency progressively approaches the oracle limit as $N$ increases. Both curves remain above the bound of Eq.~\eqref{eq:Mstar}, here evaluated from first principles rather than fitted to the data, confirming that Theorem 1 holds once $N$ exceeds the regime of validity of the bound. This is consistent with the predicted scaling $M^{\ast}\sim N^{2/3}$ and the corresponding efficiency deficit $1-\eta_{\max}\sim N^{-1/3}$. These results demonstrate that the optimal strategy assigns a nonzero but sublinear fraction of the available resources to learning, reserving the majority of the copies for thermodynamic exploitation.

\textit{Conclusions.} – 
As quantum systems of increasingly large Hilbert-space dimension become experimentally accessible, the efficient use of measurement resources, and of the information they provide, is an increasingly pressing challenge, one that the present work directly addresses. We have introduced the PIFR quantum Maxwell's demon as a minimal operational framework in which the thermodynamic cost of learning a quantum state from finitely many copies is made explicit. Our central result, Theorem 1, is a universal closed-form tradeoff that places the ergotropic efficiency between the copy-counting ceiling $1-M/N$ and a floor controlled by the product of the reconstruction precision $\varepsilon_M$ and the ergotropic susceptibility $\Lambda(\rho,H)$. Our numerical experiments, based on standard quantum state tomography, confirm the general trend predicted by this bound: a non-monotonic dependence of the efficiency on $M$, an optimal allocation $M^*$ growing sublinearly with $N$, and a corresponding approach to the oracle limit as $N$ increases [Fig.~\ref{fig:figure3}(c)]. Standard QST is, however, only one possible reconstruction strategy, and not the most efficient one. Indeed, task-adapted approaches such as threshold quantum state tomography are known to perform at least as well, and typically better, extracting a larger thermodynamic yield from the same number of copies \cite{binosi2024tailor, caruccio2025experimental}. The margin for improvement is largest for the uniform-purity ensemble, where the gap from the oracle limit remains most significant, a direction we plan to pursue.

Several other directions emerge. A fully multi-parameter finite-sample estimation theory, bridging the non-asymptotic bounds of Ref.\cite{meyer2025quantum} and the asymptotic Holevo Cram\'er--Rao theory, would sharpen the prefactor $A$ in Eq.~(15) to its state-specific value [14]. Combining our framework with Hamiltonian learning would cover the case of imperfectly known $H$, along the lines of Ref. \cite{pagliaro2026certifying}. Replacing local extraction by collective unitaries on the $N-M$ unmeasured copies~\cite{alicki2013entanglement,andolina2019extractable,campaioli2024colloquium} could probe whether coherent extraction strategies modify the $N^{-1/3}$ rate. Finally, extending the setup to incoherent or observational notions of ergotropy~\cite{francica2020quantum} would connect the present tradeoff to the resource theory of coherence~\cite{baumgratz2014quantifying}.

\textit{Acknowledgments. --}
G.Guar. kindly acknowledges support from the Ministero dell’Università e della Ricerca (MUR) under the “Rita Levi-Montalcini” grant and to INFN. G.Guar. is also grateful to Emanuele Tumbiolo for insightful discussions. 

\textit{Authors' Contributions. --}
G. Guar. and M. L. conceived the idea and introduced the framework and the key quantities. G. Guar. derived the finite-sample tradeoff bounds and the theoretical results of the supplemental material. All the authors contributed to the numerical simulations. G. Guar. and M. L. wrote the manuscript. 

\bibliography{refs}

\widetext

\appendix

\input{supp_mat}

\end{document}

%% file: supp_mat.tex
\begin{center}
\textbf{\large Supplemental Material:\\
Quantum work extraction from partial information and with finite resources}
\end{center}

\setcounter{equation}{0}
\setcounter{figure}{0}
\setcounter{table}{0}
\setcounter{page}{1}
\makeatletter
\renewcommand{\thesection}{S\arabic{section}}
\renewcommand{\thesubsection}{\thesection.\arabic{subsection}}
\renewcommand{\theequation}{S\arabic{equation}}
\renewcommand{\thefigure}{S\arabic{figure}}

\providecommand{\norm}[1]{\left\lVert #1\right\rVert}
\providecommand{\abs}[1]{\left\lvert #1\right\rvert}
\providecommand{\ketbra}[2]{\left|#1\right\rangle\!\left\langle #2\right|}
\renewcommand{\Tr}{\operatorname{Tr}}

\newtheorem{lemma}{Lemma}
\renewcommand{\thelemma}{S\arabic{lemma}}
\newtheorem{remark}{Remark}
\renewcommand{\theremark}{S\arabic{remark}}

\noindent In this Supplemental Material we collect the technical material behind the results of the main text. Section A fixes notation and conventions and explains the origin of the quantity $\Lambda(\rho,H)$ that we call the ergotropic susceptibility. Section B contains the proof of Lemma 1. Section C derives closed-form expressions for the reconstruction precision $\varepsilon_M$, both from the finite-sample metrological bounds of Meyer \emph{et al.}~\cite{Meyer2025} and from elementary concentration inequalities; there we also discuss in detail the multi-parameter aspects of the problem and their relation to the Holevo Cram\'er--Rao bound. Section D derives the optimal splitting $M^{\ast}\propto N^{2/3}$ of the copy budget and the resulting $N^{-1/3}$ approach to the oracle limit. Section E discusses the sharpness and the scope of the various inequalities, and derives a refinement of Lemma 1 based on eigenvector perturbation theory. Section F quantifies the regime of validity of the main-text bound. Throughout, $\log$ denotes the natural logarithm, and, as in the main text, $\rhoI$ denotes the state reconstructed from the dataset
collected with $M$ copies and measurement design $\Pi$; similarly for $\WI $ and $\etaI $.

\section{A --- Setup, notation, and the ergotropic susceptibility}
\label{sec:setup}

We consider a quantum system with Hilbert space of finite dimension $D$, Hamiltonian
\begin{equation}
H=\sum_{k=1}^{D}\epsilon_{k}\ketbra{\epsilon_{k}}{\epsilon_{k}},\qquad
\epsilon_{1}\le\epsilon_{2}\le\dots\le\epsilon_{D},
\label{eq:S-H}
\end{equation}
and unknown state $\rho$ with spectral decomposition
\begin{equation}
\rho=\sum_{k=1}^{D}r_{k}\ketbra{r_{k}}{r_{k}},\qquad
r_{1}\ge r_{2}\ge\dots\ge r_{D}\ge 0 .
\label{eq:S-rho}
\end{equation}
The spectral spread of the Hamiltonian is $\Delta_{H}:=\epsilon_{D}-\epsilon_{1}$. For a Hermitian operator $A$ with eigenvalues $\lambda_{j}(A)$ we use the Schatten norms
\begin{equation}
\norm{A}_{1}=\sum_{j}\abs{\lambda_{j}(A)},\qquad
\norm{A}_{2}=\Bigl(\sum_{j}\lambda_{j}(A)^{2}\Bigr)^{1/2},\qquad
\norm{A}_{\infty}=\max_{j}\abs{\lambda_{j}(A)} ,
\label{eq:S-norms}
\end{equation}
and for two states the trace distance and the Uhlmann fidelity,
\begin{equation}
\D(\rho,\sigma):=\tfrac12\norm{\rho-\sigma}_{1}\in[0,1],
\qquad
\F(\rho,\sigma):=\Tr\sqrt{\sqrt{\sigma}\,\rho\,\sqrt{\sigma}}\in[0,1] .
\label{eq:S-distances}
\end{equation}

The protocol analysed in the main text rests on four structural assumptions, which we state here once: (A1) the demon holds $N$ independent, identically prepared copies $\rho^{\otimes N}$; (A2) measurements are destructive, so the $M<N$ copies used for reconstruction are lost for work extraction; (A3) the Hamiltonian \eqref{eq:S-H} is known exactly; (A4) work is extracted locally, by applying the same single-copy unitary $\UWI :=U_{W}(\rhoI)$, computed from the reconstructed state, to each of the remaining $N-M$ copies. Two further hypotheses of a more technical nature---a confidence guarantee on the reconstruction, and non-degeneracy of the spectrum of $\rho$---are needed only locally, and we introduce them where they first enter (Secs.~\ref{sec:epsilonM} and \ref{sec:limitations}).

A central role in Theorem~1 of the main text is played by the dimensionless combination
\begin{equation}
\Lambda(\rho,H):=\frac{2\Delta_{H}}{\Werg(\rho,H)} ,
\label{eq:S-Lambda}
\end{equation}
which we refer to as the ergotropic susceptibility. The name is chosen in analogy with linear-response theory. A susceptibility is a state-dependent coefficient converting a small perturbation into the leading-order response of an observable, as the magnetic susceptibility $\chi=\partial M/\partial H|_{H=0}$ converts a weak applied field into a change of magnetization. Lemma~\ref{lem:decomp} of the main text can be written as
\begin{equation}
\frac{\Delta W}{\Werg(\rho,H)}\;\le\;\Lambda(\rho,H)\,\D(\rho,\rhoI) ,
\label{eq:S-susc}
\end{equation}
so that $\Lambda$ plays precisely this role: it is (an upper bound on) the linear-response coefficient of the relative work deficit with respect to the reconstruction error, measured in trace distance. States with $\Lambda\gg1$ are fragile---a small error in the reconstruction destroys a large fraction of the extractable work, as happens for nearly passive states, where $\Werg\to0$ at fixed $\Delta_H$---whereas states with $\Lambda$ of order one are robust. Equivalently, as shown in Sec.~\ref{sec:regime}, $1/\Lambda(\rho,H)$ is the largest reconstruction error compatible with a nontrivial guarantee on the efficiency.

\section{B --- Proof of Lemma 1}
\label{sec:proof-lemma}

We prove the ergotropic-deficit bound
\begin{equation}
\Delta W:=\Werg(\rho,H)-\WI\;\le\;2\Delta_{H}\,\D(\rho,\rhoI) .
\label{eq:S-lemma1}
\end{equation}
Four density matrices enter the argument, and it is worth distinguishing them carefully at the outset: the true state $\rho$ of Eq.~\eqref{eq:S-rho}; the reconstructed state $\rhoI=\sum_{k}\rrec{k}\ketbra{\rrec{k}}{\rrec{k}}$, with $\rrec{1}\ge \rrec{2}\ge\dots$; the passive state of the true state,
\begin{equation}
\sigma_{\star}:=U_{W}(\rho)\,\rho\,U_{W}(\rho)^{\dagger}=\sum_{k}r_{k}\ketbra{\epsilon_{k}}{\epsilon_{k}} ;
\label{eq:S-sigmastar}
\end{equation}
and the state actually produced by the demon,
\begin{equation}
\sigmaI :=\UWI \,\rho\,\UWIdag  .
\label{eq:S-sigmaI}
\end{equation}
Note that $\sigmaI $ is built from the true state $\rho$, not from $\rhoI$: the demon applies its inferred unitary to actual copies of $\rho$. The mismatch between $\sigmaI $ and the passive state is the whole source of the deficit.

\begin{lemma}[Two-term decomposition]\label{lem:decomp}
With the notation above,
\begin{equation}
\Delta W=\underbrace{\Tr\!\bigl[\UWIdag H\UWI \,(\rho-\rhoI)\bigr]}_{\mathrm{(a)}}
+\underbrace{\bigl[\Emin(\rhoI)-\Emin(\rho)\bigr]}_{\mathrm{(b)}} .
\label{eq:S-twoterm}
\end{equation}
\end{lemma}

\begin{proof}
From Eqs.~(2) and (7) of the main text, $\Werg(\rho,H)=\Tr[\rho H]-\Emin(\rho)$ and $\WI=\Tr[\rho H]-\Tr[H\sigmaI ]$, so that
\begin{equation}
\Delta W=\Tr[H\sigmaI ]-\Emin(\rho) .
\label{eq:S-DW-first}
\end{equation}
Adding and subtracting $\Emin(\rhoI)$,
\begin{equation}
\Delta W=\bigl\{\Tr[H\sigmaI ]-\Emin(\rhoI)\bigr\}+\bigl\{\Emin(\rhoI)-\Emin(\rho)\bigr\} .
\label{eq:S-DW-split}
\end{equation}
The second brace is term (b). For the first, recall that $\UWI =U_{W}(\rhoI)$ is, by construction, the unitary that minimises the final energy \emph{of the reconstructed state}: by Eqs.~(3)--(4) of the main text,
\begin{equation}
\Tr\!\bigl[H\,\UWI \,\rhoI\,\UWIdag \bigr]=\Emin(\rhoI) .
\label{eq:S-Emin-rhoI}
\end{equation}
The argument of the conjugation in Eq.~\eqref{eq:S-Emin-rhoI} must be $\rhoI$ and not $\rho$; $\UWI $ is optimal for the state the demon believes it holds, not for the state it actually holds. Substituting Eqs.~\eqref{eq:S-sigmaI} and \eqref{eq:S-Emin-rhoI} into the first brace of Eq.~\eqref{eq:S-DW-split} and using cyclicity of the trace,
\begin{equation}
\Tr[H\sigmaI ]-\Emin(\rhoI)
=\Tr\!\bigl[H \UWI \,\rho\,\UWIdag \bigr]-\Tr\!\bigl[H \UWI \,\rhoI\,\UWIdag \bigr]
=\Tr\!\bigl[\UWIdag H\UWI \,(\rho-\rhoI)\bigr],
\end{equation}
which is term (a).
\end{proof}

Neither term in Eq.~\eqref{eq:S-twoterm} is separately of definite sign, but their sum is non-negative because $\sigma_{\star}$ minimises $\Tr[HU\rho U^{\dagger}]$ over all unitaries; this recovers Eq.~(9) of the main text.

\begin{lemma}[Lipschitz continuity of the minimum energy]\label{lem:lipschitz}
Let $\Emin(\varrho):=\min_{U\in\mathcal{U}(D)}\Tr[U\varrho\, U^{\dagger}H]$, as in Eq.~(3) of the main text. Then, for any two states $\rho_{1},\rho_{2}$,
\begin{equation}
\abs{\Emin(\rho_{1})-\Emin(\rho_{2})}\le\Delta_{H}\,\D(\rho_{1},\rho_{2}) .
\label{eq:S-lipsch}
\end{equation}
\end{lemma}

\begin{proof}
Let $U_{j}$ achieve the minimum for $\rho_{j}$. Using $U_{2}$ as a trial unitary for $\rho_{1}$,
\begin{equation}
\Emin(\rho_{1})-\Emin(\rho_{2})
\le\Tr[U_{2}\rho_{1}U_{2}^{\dagger}H]-\Tr[U_{2}\rho_{2}U_{2}^{\dagger}H]
=\Tr\!\bigl[A\,(\rho_{1}-\rho_{2})\bigr],\qquad A:=U_{2}^{\dagger}HU_{2} .
\label{eq:S-lipsch-step}
\end{equation}
The operator $A$ has the same spectrum as $H$, so that $\norm{A-\alpha\mathbb{I}}_{\infty}=\Delta_{H}/2$ for the centring $\alpha=(\epsilon_{1}+\epsilon_{D})/2$. Since $\rho_{1}-\rho_{2}$ is traceless, we may substitute $A\to A-\alpha\mathbb{I}$ in Eq.~\eqref{eq:S-lipsch-step} at no cost, and H\"older's inequality $\abs{\Tr[XY]}\le\norm{X}_{\infty}\norm{Y}_{1}$ gives
\begin{equation}
\Emin(\rho_{1})-\Emin(\rho_{2})\le\tfrac{\Delta_{H}}{2}\norm{\rho_{1}-\rho_{2}}_{1}=\Delta_{H}\,\D(\rho_{1},\rho_{2}) .
\end{equation}
Exchanging the roles of $\rho_{1}$ and $\rho_{2}$ yields the reverse inequality.
\end{proof}

The proof of Eq.~\eqref{eq:S-lemma1} is now immediate. Term (a) of Lemma~\ref{lem:decomp} involves the operator $\UWIdag H\UWI $, again isospectral to $H$; the same centring-plus-H\"older argument used in Lemma~\ref{lem:lipschitz} gives $\abs{\mathrm{(a)}}\le\Delta_{H}\,\D(\rho,\rhoI)$. Term (b) is bounded by Lemma~\ref{lem:lipschitz} with $\rho_{1}=\rhoI$, $\rho_{2}=\rho$, giving $\abs{\mathrm{(b)}}\le\Delta_{H}\,\D(\rho,\rhoI)$. Adding the two contributions proves Lemma~\ref{lem:decomp}. \hfill$\square$

\begin{remark}
Equation~\eqref{eq:S-lemma1} may be regarded as the ergotropic analogue of the Fannes--Audenaert continuity bound for the von Neumann entropy~\cite{Fannes1973,Audenaert2007}: both control the variation of a spectral functional under trace-norm perturbations, with a constant fixed by a global scale ($\Delta_{H}$ here, $\log D$ there). As for Fannes-type bounds, the linear dependence on $\D$ is tight up to a constant for generic spectra but degrades near spectral degeneracies of $\rho$; a refinement adapted to well-gapped spectra is given in Sec.~\ref{sec:limitations}.
\end{remark}

\section{C --- The reconstruction precision \texorpdfstring{$\varepsilon_M$}{eps\_M}}
\label{sec:epsilonM}

Lemma~\ref{lem:decomp} reduces the thermodynamic problem to an inference problem: bounding the trace distance $\D(\rho,\rhoI)$ attainable with $M$ copies. We formalise the interface between the two problems through the following hypothesis on the reconstruction strategy $\mathcal{R}$:
\begin{quote}
(A5) For every failure probability $p_{\mathrm{fail}}\in(0,1)$ there exists $\varepsilon_{M}=\varepsilon_{M}(\mathcal{R},p_{\mathrm{fail}})$ such that $\D(\rho,\rhoI)\le\varepsilon_{M}$ with probability at least $1-p_{\mathrm{fail}}$ over the measurement outcomes.
\end{quote}
Theorem~1 of the main text holds verbatim for any strategy obeying (A5), whatever the value of $\varepsilon_{M}$; the purpose of this section is to compute what $\varepsilon_{M}$ can be. We give two complementary answers. The first (Secs.~\ref{sec:FvdG}--\ref{sec:multi-MF}) starts from the finite-sample metrological bounds of Ref.~\cite{Meyer2025} and produces the fundamental, strategy-independent floor for $\varepsilon_M$; since full state reconstruction is a $(D^{2}-1)$-parameter estimation problem, this requires some care with the incompatibility issues characteristic of multi-parameter quantum metrology, which we discuss against the Holevo Cram\'er--Rao theory~\cite{Helstrom1976,Holevo1976,YangChiribellaHayashi2019,AlbarelliFrielDatta2019,BelliardoGiovannetti2021}. The second answer (Sec.~\ref{sec:BHC}) uses only classical concentration inequalities and applies to the concrete tomographic protocols employed in the main text~\cite{Pagliaro2026,Zambrano2026,Mortimer2026}; it is weaker in its dimensional prefactors but fully explicit at finite $M$.

\subsection{From fidelity to trace distance}
\label{sec:FvdG}

The Fuchs--van de Graaf inequalities~\cite{Fuchs1999} state that
\begin{equation}
1-\F(\rho,\sigma)\;\le\;\D(\rho,\sigma)\;\le\;\sqrt{1-\F(\rho,\sigma)^{2}} .
\label{eq:S-FvdG}
\end{equation}
Since $1-\F^{2}=(1-\F)(1+\F)\le2(1-\F)$, the upper bound implies
\begin{equation}
\D(\rho,\rhoI)^{2}\;\le\;2\,\bigl[1-\F(\rho,\rhoI)\bigr] ,
\label{eq:S-D-from-F}
\end{equation}
an exact inequality, valid for all values of the fidelity. Equation~\eqref{eq:S-D-from-F} converts any infidelity guarantee into a trace-distance guarantee, and is the only bridge between the two metrics that we shall need.

\subsection{Finite-sample bound for a single parameter}
\label{sec:single-MF}

Consider first a smooth one-parameter family $t\mapsto\rho(t)$. An estimation protocol is a POVM $\{Q(\tau)\}$ acting on $\rho(t)^{\otimes M}$ and returning an estimate $\tau$; for a tolerance $\delta>0$ its (minimax) success probability is
\begin{equation}
\eta(\delta,Q):=\inf_{t}\int d\tau\, w_{\delta}(t-\tau)\,\Tr\!\bigl[\rho(t)^{\otimes M}Q(\tau)\bigr],
\qquad w_{\delta}(x)=\mathbf{1}_{\{\abs{x}\le\delta\}} ,
\label{eq:S-successprob}
\end{equation}
and $\eta^{\ast}(\delta):=\sup_{Q}\eta(\delta,Q)$. Corollary~2 of Ref.~\cite{Meyer2025} provides the fundamental limit
\begin{equation}
1-\eta^{\ast}(\delta)\;\ge\;\tfrac14\,\sup_{\abs{t-t'}>2\delta}\F\bigl(\rho(t),\rho(t')\bigr)^{2M} .
\label{eq:S-MF-cor2}
\end{equation}
Writing $p:=1-\eta^{\ast}(\delta)$ for the failure probability and taking logarithms of Eq.~\eqref{eq:S-MF-cor2},
\begin{equation}
\log\frac{1}{4p}\;\le\;2M\inf_{\abs{t-t'}>2\delta}\bigl[-\log\F\bigl(\rho(t),\rho(t')\bigr)\bigr] .
\label{eq:S-MF-log}
\end{equation}
The quantity $-\log\F$ is controlled, for small separations, by the quantum Fisher information (QFI). With the standard convention $\partial_{t}\rho=\tfrac12\{L,\rho\}$, $\mathcal{F}_{Q}=\Tr[\rho L^{2}]$, the fidelity between neighbouring states expands as~\cite{BraunsteinCaves1994,Hubner1992}
\begin{equation}
\F\bigl(\rho(t),\rho(t+\tau)\bigr)=1-\frac{\mathcal{F}_{Q}(t)}{8}\,\tau^{2}+\mathcal{O}(\tau^{3})
\qquad\Longrightarrow\qquad
-\log\F=\frac{\mathcal{F}_{Q}(t)}{8}\,\tau^{2}+\mathcal{O}(\tau^{3}) .
\label{eq:S-F-expansion}
\end{equation}
Since $\F$ decreases with the separation, the infimum in Eq.~\eqref{eq:S-MF-log} is attained at the boundary $\abs{t-t'}=2\delta$; retaining the leading term of Eq.~\eqref{eq:S-F-expansion}---the validity of this truncation is discussed in Sec.~\ref{sec:limitations}---one finds $-\log\F=\mathcal{F}_{Q}\,\delta^{2}/2$, hence $\log(1/4p)\le M\mathcal{F}_{Q}\delta^{2}$ and
\begin{equation}
\delta\;\ge\;\delta^{\ast}_{M}(p):=\sqrt{\frac{\log[1/(4p)]}{M\,\mathcal{F}_{Q}}}\, .
\label{eq:S-MF-single}
\end{equation}
Equation~\eqref{eq:S-MF-single} is a converse (impossibility) statement: no protocol can guarantee a tolerance below $\delta^{\ast}_{M}(p)$ at confidence $1-p$. Reference~\cite{Meyer2025} also establishes achievability of this scaling, with modified constants, so we shall use $\delta^{\ast}_{M}(p)$ as the benchmark tolerance of an optimal single-parameter strategy; the caveats attached to this reading in the multi-parameter setting are spelled out below.

\subsection{The multi-parameter problem: union bound versus Holevo bound}
\label{sec:multi-MF}

Full state reconstruction is the simultaneous estimation of $D^{2}-1$ real parameters, say generalised Bloch coordinates $\bm{t}\in\mathbb{R}^{D^{2}-1}$, with quantum Fisher information matrix (QFIM)
\begin{equation}
\bigl[\mathcal{F}_{Q}\bigr]_{ij}
=\Tr\!\Bigl[\rho\,\frac{L_{i}L_{j}+L_{j}L_{i}}{2}\Bigr]
=2\sum_{k,l\,:\,p_{k}+p_{l}>0}\frac{\mathrm{Re}\bigl[\bra{k}\partial_{i}\rho\ket{l}\bra{l}\partial_{j}\rho\ket{k}\bigr]}{p_{k}+p_{l}} ,
\label{eq:S-QFIM}
\end{equation}
where $\rho=\sum_{k}p_{k}\ketbra{k}{k}$ and $\partial_{i}\rho=\tfrac12\{L_{i},\rho\}$. We assume $\rho$ of full rank, so that the model is regular and $\mathcal{F}_{Q}\succ0$; rank-deficient states can be treated by restriction to their support. The multi-parameter generalisation of Eq.~\eqref{eq:S-F-expansion} reads
\begin{equation}
1-\F\bigl(\rho(\bm t),\rho(\bm t+d\bm t)\bigr)=\tfrac18\,d\bm t^{\top}\mathcal{F}_{Q}\,d\bm t+\mathcal{O}(\abs{d\bm t}^{3}) ,
\label{eq:S-F-QFIM}
\end{equation}
which is the statement that the QFIM is (four times) the metric tensor of the Bures geometry, $d_{B}^{2}=2(1-\F)$ and $d_{B}^{2}=\tfrac14 d\bm t^{\top}\mathcal{F}_{Q}d\bm t$ infinitesimally~\cite{BraunsteinCaves1994,Hubner1992}.

It is a well-known feature of multi-parameter quantum estimation that the parameters cannot, in general, be treated independently: the optimal measurements for different parameters need not commute, and the single-parameter bound~\eqref{eq:S-MF-single} cannot simply be saturated along every direction at once with the same copies~\cite{Helstrom1976,Holevo1976}. We therefore proceed in two stages: first a derivation based on the union bound, which is elementary and delivers a fully non-asymptotic guarantee, and then a comparison with the Holevo Cram\'er--Rao bound, which is the sharp asymptotic theory and quantifies exactly what the elementary route gives away.

Diagonalise the QFIM, $\mathcal{F}_{Q}=\mathrm{diag}(\lambda_{1},\dots,\lambda_{D^{2}-1})$ with $\lambda_{i}>0$, and take the parameters $t_{i}$ along its eigendirections, so that the QFI of the one-dimensional subfamily along direction $i$ is $\lambda_{i}$. Suppose the strategy meets the single-parameter benchmark~\eqref{eq:S-MF-single} along each direction, with an individual failure probability $p_{i}$:
\begin{equation}
\Pr\bigl[\abs{\hat t_{i}-t_{i}}>\delta_{i}\bigr]\le p_{i},
\qquad
\delta_{i}=\sqrt{\frac{\log[1/(4p_{i})]}{M\lambda_{i}}}\, .
\label{eq:S-single-dir}
\end{equation}
This is the one genuinely non-trivial hypothesis of the derivation: it holds exactly when the symmetric logarithmic derivatives $L_{i}$ commute (as for pure states, or for classically correlated families), and it holds asymptotically in $M$ for arbitrary states, because collective measurements on $\rho^{\otimes M}$ attain the Holevo bound~\cite{YangChiribellaHayashi2019} and with it, up to a state-dependent factor discussed below, the per-direction SLD performance. We flag it explicitly as hypothesis (H) and return to it at the end of the subsection.

By the union bound (Boole's inequality), which for any events $A_{1},\dots,A_{n}$ states
\begin{equation}
\Pr\Bigl[\,\bigcup_{i=1}^{n}A_{i}\Bigr]\le\sum_{i=1}^{n}\Pr[A_{i}] ,
\label{eq:S-union}
\end{equation}
applied to $A_{i}=\{\abs{\hat t_{i}-t_{i}}>\delta_{i}\}$, all $D^{2}-1$ estimates are simultaneously within tolerance with probability at least $1-\sum_{i}p_{i}$. Splitting the failure budget uniformly,
\begin{equation}
p_{i}=\frac{p_{\mathrm{fail}}}{D^{2}-1}
\qquad\Longrightarrow\qquad
\delta_{i}^{2}=\frac{\log\!\bigl[(D^{2}-1)/(4p_{\mathrm{fail}})\bigr]}{M\lambda_{i}}\quad\text{for all }i ,
\label{eq:S-deltai}
\end{equation}
jointly with probability at least $1-p_{\mathrm{fail}}$. The uniform split is not merely convenient: on the joint success event the relevant error functional is, by Eq.~\eqref{eq:S-F-QFIM},
\begin{equation}
1-\F(\rho,\rhoI)\;\le\;\tfrac18\sum_{i=1}^{D^{2}-1}\lambda_{i}\,(\hat t_{i}-t_{i})^{2}
\;\le\;\tfrac18\sum_{i=1}^{D^{2}-1}\lambda_{i}\,\delta_{i}^{2}
\;=\;\tfrac18\sum_{i=1}^{D^{2}-1}\frac{\log[1/(4p_{i})]}{M} ,
\label{eq:S-1mF}
\end{equation}
in which the QFIM eigenvalues have cancelled identically. Minimising $\sum_{i}\log(1/p_{i})$ over the allocation subject to $\sum_{i}p_{i}=p_{\mathrm{fail}}$ gives, by symmetry (or a one-line Lagrange computation, $\partial_{p_{i}}[-\log p_{i}+\mu p_{i}]=0\Rightarrow p_{i}=1/\mu$), precisely the uniform allocation~\eqref{eq:S-deltai}, which is therefore optimal within this scheme. Combining Eqs.~\eqref{eq:S-1mF}, \eqref{eq:S-deltai} and the Fuchs--van de Graaf bridge~\eqref{eq:S-D-from-F},
\begin{equation}
\D(\rho,\rhoI)^{2}\le2\bigl[1-\F(\rho,\rhoI)\bigr]\le\frac{(D^{2}-1)\log\!\bigl[(D^{2}-1)/(4p_{\mathrm{fail}})\bigr]}{4M} ,
\end{equation}
that is,
\begin{equation}
\varepsilon_{M}\;=\;\frac12\,\sqrt{\frac{(D^{2}-1)\,\log\!\bigl[(D^{2}-1)/(4p_{\mathrm{fail}})\bigr]}{M}}\, ,
\label{eq:S-epsilonM}
\end{equation}
valid with confidence $1-p_{\mathrm{fail}}$ and to leading order in $1/M$. Together with Theorem~1 of the main text this gives the benchmark lower envelope
\begin{equation}
\etaI\;\ge\;\Bigl(1-\frac{M}{N}\Bigr)
\biggl[\,1-\frac{\Lambda(\rho,H)}{2}\sqrt{\frac{(D^{2}-1)\log[(D^{2}-1)/(4p_{\mathrm{fail}})]}{M}}\,\biggr]_{+} .
\label{eq:S-envelope-MF}
\end{equation}

We now confront this elementary derivation with the sharp theory. In multi-parameter quantum estimation the attainable asymptotic precision is governed not by the SLD quantities entering Eq.~\eqref{eq:S-single-dir} but by the Holevo Cram\'er--Rao bound~\cite{Holevo1976}: for any weight matrix $W\succeq0$ and any (possibly collective, entangled) measurement on $\rho^{\otimes M}$ with locally unbiased estimator $\hat{\bm t}$,
\begin{equation}
M\,\Tr\!\bigl[W\,\mathrm{Cov}(\hat{\bm t})\bigr]\;\ge\;
C^{H}(W,\rho):=\min_{\{X_{i}\}}
\Bigl\{\Tr\!\bigl[W^{1/2}\,\mathrm{Re}\,Z\,W^{1/2}\bigr]
+\bigl\|W^{1/2}\,\mathrm{Im}\,Z\,W^{1/2}\bigr\|_{1}\Bigr\} ,
\label{eq:S-Holevo}
\end{equation}
where $Z_{ij}:=\Tr[\rho X_{i}X_{j}]$ and the minimisation runs over Hermitian collections satisfying the local-unbiasedness constraints $\Tr[\partial_{i}\rho\,X_{j}]=\delta_{ij}$. The Holevo functional is sandwiched between the SLD value and twice it,
\begin{equation}
\Tr\!\bigl[W\mathcal{F}_{Q}^{-1}\bigr]\;\le\;C^{H}(W,\rho)\;\le\;2\,\Tr\!\bigl[W\mathcal{F}_{Q}^{-1}\bigr] ,
\label{eq:S-Holevo-sandwich}
\end{equation}
with equality on the left if and only if the SLD-incompatibility obstruction vanishes (in particular for pure states or commuting SLDs), and the factor $2$ on the right being the universal worst case~\cite{AlbarelliFrielDatta2019,BelliardoGiovannetti2021}. Moreover, $C^{H}$ is attainable in the limit of collective measurements on asymptotically many copies~\cite{YangChiribellaHayashi2019}, which is what singles it out as \emph{the} multi-parameter benchmark.

For our problem the natural weight is $W=\mathcal{F}_{Q}$, since by Eq.~\eqref{eq:S-F-QFIM} the Bures infidelity is precisely the $\mathcal{F}_{Q}$-weighted quadratic error. Taking expectations of Eq.~\eqref{eq:S-F-QFIM} and using Eq.~\eqref{eq:S-Holevo},
\begin{equation}
M\;\mathbb{E}\bigl[1-\F(\rho,\rhoI)\bigr]\;\simeq\;\frac{M}{8}\,\Tr\!\bigl[\mathcal{F}_{Q}\,\mathrm{Cov}(\hat{\bm t})\bigr]
\;\ge\;\frac{C^{H}(\mathcal{F}_{Q},\rho)}{8}
\;\ge\;\frac{D^{2}-1}{8} ,
\label{eq:S-Holevo-infid}
\end{equation}
where the last step used $\Tr[\mathcal{F}_{Q}\mathcal{F}_{Q}^{-1}]=D^{2}-1$. Defining the incompatibility ratio of the tomographic model
\begin{equation}
R(\rho):=\frac{C^{H}(\mathcal{F}_{Q},\rho)}{D^{2}-1}\;\in\;[1,2] ,
\label{eq:S-Rratio}
\end{equation}
the asymptotically optimal collective strategy achieves $\mathbb{E}[1-\F]=R(\rho)(D^{2}-1)/(8M)\,[1+o(1)]$, and hence, through Eq.~\eqref{eq:S-D-from-F}, a typical trace-distance error
\begin{equation}
\varepsilon_{M}^{\mathrm{Hol}}\;=\;\Theta\!\left(\sqrt{\frac{R(\rho)\,(D^{2}-1)}{M}}\,\right) .
\label{eq:S-Holevo-rate}
\end{equation}

The comparison of Eqs.~\eqref{eq:S-epsilonM} and \eqref{eq:S-Holevo-rate} is instructive on three counts. First, the two routes agree on the scaling $\sqrt{(D^{2}-1)/M}$, which is what feeds into the optimisation of Sec.~\ref{sec:optimal}; the union-bound route pays an extra $\sqrt{\log(D^{2}-1)}$, which is the familiar and, in general, unavoidable price of upgrading an expectation-value statement to an explicit high-confidence tail bound (the same logarithm appears in the sample-complexity analyses of quantum tomography~\cite{HHJLW2017}). Second, the union bound is agnostic about measurement incompatibility, which in the sharp theory surfaces as the factor $R(\rho)\le2$ of Eq.~\eqref{eq:S-Rratio}; at the level of Eq.~\eqref{eq:S-epsilonM} this uncertainty is invisible because it affects the constant, not the scaling. Third, the Holevo rate~\eqref{eq:S-Holevo-rate} is an asymptotic statement about expectations, and turning it into a finite-$M$, high-probability guarantee of the form (A5)---with the correct constant $R(\rho)$ and without the logarithm---remains, to the best of our knowledge, an open problem in multi-parameter quantum metrology, the finite-sample techniques of Ref.~\cite{Meyer2025} being at present fully developed only in the single-parameter case. In this light Eq.~\eqref{eq:S-epsilonM} should be read as the sharpest \emph{non-asymptotic} benchmark currently derivable, correct in its dimensional scaling, and conservative by at most a logarithm and an $\mathcal{O}(1)$ incompatibility factor.

Two final remarks on hypothesis (H). One could bypass it altogether by splitting the copies into $D^{2}-1$ disjoint batches of $M/(D^{2}-1)$ and estimating one direction per batch; Eq.~\eqref{eq:S-single-dir} then holds by construction, at the price of replacing $M\to M/(D^{2}-1)$ throughout, i.e.\ an extra factor $\sqrt{D^{2}-1}$ in Eq.~\eqref{eq:S-epsilonM}. This sequential scheme is strictly wasteful---collective strategies provably beat it asymptotically~\cite{YangChiribellaHayashi2019}---but it shows that even the weakest reading of our assumptions degrades the bound only polynomially in $D$, leaving the $M^{-1/2}$ dependence, and with it all the conclusions of Sec.~\ref{sec:optimal}, untouched. Conversely, for the states for which the SLDs commute, (H) holds at finite $M$ and Eq.~\eqref{eq:S-epsilonM} is an honest finite-sample statement with no hidden constants.

\subsection{Concentration bounds for laboratory protocols}
\label{sec:BHC}

The bound of the previous subsection is a fundamental benchmark, independent of the reconstruction strategy. For the specific protocols implemented in the main text (Pauli-basis quantum state tomography and its threshold variant) an entirely elementary route, based on classical concentration, yields explicit certificates of the type (A5); this is the approach taken, in the context of ergotropy and many-body property certification, in Refs.~\cite{Pagliaro2026,Zambrano2026,Mortimer2026}, whose multinomial version is the Bretagnolle--Huber--Carol inequality~\cite{BHC}.

Suppose the reconstruction measures $K$ observables $O_{1},\dots,O_{K}$ with spectra in $[-1,1]$, devoting $M_{i}$ copies to $O_{i}$, $\sum_{i}M_{i}=M$, and let $\hat o_{i}$ be the empirical mean of the outcomes for $O_{i}$. Hoeffding's inequality~\cite{Hoeffding1963} for i.i.d.\ variables bounded in $[-1,1]$ gives, for each $i$ and every $s>0$,
\begin{equation}
\Pr\bigl[\abs{\hat o_{i}-\Tr(\rho O_{i})}\ge s\bigr]\le2\,e^{-M_{i}s^{2}/2} .
\label{eq:S-Hoeffding}
\end{equation}
Setting the right-hand side equal to $p_{\mathrm{fail}}/K$, i.e.\ $s_{i}=\sqrt{2\log(2K/p_{\mathrm{fail}})/M_{i}}$, and applying the union bound~\eqref{eq:S-union} over the $K$ settings,
\begin{equation}
\abs{\hat o_{i}-\Tr(\rho O_{i})}\le\sqrt{\frac{2\log(2K/p_{\mathrm{fail}})}{M_{i}}}
\qquad\text{for all }i\text{ simultaneously, with probability}\ge1-p_{\mathrm{fail}} .
\label{eq:S-Hoeffding-joint}
\end{equation}
For $n$-qubit Pauli tomography one has $D=2^{n}$, $K=3^{n}$ settings and, with the uniform allocation used in the main text, $M_{i}=M/3^{n}$. Denote by $\hat{\bm o}$ and $\bm o$ the vectors of estimated and true expectation values, and by $\rhoI=\mathcal{L}(\hat{\bm o})$ the linear-inversion estimate, with $\mathcal{L}$ the (linear) reconstruction map of the informationally complete frame. Then $\rhoI-\rho=\mathcal{L}(\hat{\bm o}-\bm o)$, and denoting by $\kappa$ the operator norm of $\mathcal{L}$ from $\ell^{2}(\mathbb{R}^{K})$ to Hilbert--Schmidt space (the inverse smallest singular value of the frame operator),
\begin{equation}
\norm{\rhoI-\rho}_{2}\le\kappa\,\norm{\hat{\bm o}-\bm o}_{2}\le\kappa\sqrt{K}\,\max_{i}\abs{\hat o_{i}-o_{i}} ,
\qquad
\norm{\rhoI-\rho}_{1}\le\sqrt{D}\,\norm{\rhoI-\rho}_{2} ,
\label{eq:S-inversion}
\end{equation}
the last step by Cauchy--Schwarz. Assembling Eqs.~\eqref{eq:S-Hoeffding-joint} and \eqref{eq:S-inversion} with $M_{i}=M/K$,
\begin{equation}
\D(\rho,\rhoI)=\tfrac12\norm{\rhoI-\rho}_{1}
\;\le\;\frac{\kappa}{\sqrt2}\,\sqrt{D}\;3^{n}\,
\sqrt{\frac{\log\!\bigl(2\cdot3^{n}/p_{\mathrm{fail}}\bigr)}{M}}\, ,
\label{eq:S-QST-bound}
\end{equation}
with probability at least $1-p_{\mathrm{fail}}$. We have not optimised the dimensional prefactor, which reflects the conditioning of the specific frame rather than anything fundamental; what matters is its exponential growth with $n$, to be contrasted with the $\sqrt{D^{2}-1}=\sqrt{4^{n}-1}$ of the benchmark~\eqref{eq:S-epsilonM}, and the fact that the $M^{-1/2}$ and $\sqrt{\log(1/p_{\mathrm{fail}})}$ dependences are the same. Task-adapted schemes such as threshold tomography effectively replace the number of settings $3^{n}$ by a state-dependent effective number $S(\rho)<3^{n}$ and reallocate the shots accordingly, which is the origin of their advantage in the scarce-resource regime documented in the main text; a sharp analytical characterisation of $S(\rho)$ is beyond our present scope.

Both realisations of $\varepsilon_{M}$---the fundamental benchmark~\eqref{eq:S-epsilonM} and the protocol-specific certificate~\eqref{eq:S-QST-bound}---have the form $\varepsilon_{M}\simeq A/\sqrt{M}$ and can be inserted interchangeably into Theorem~1; only the constant $A$ differs.

\section{D --- Optimal allocation of the copy budget}
\label{sec:optimal}

Inserting $\varepsilon_{M}=A/\sqrt{M}$ into the lower envelope of Theorem~1 defines
\begin{equation}
g(M):=\Bigl(1-\frac{M}{N}\Bigr)\Bigl[1-\frac{A}{\sqrt M}\Bigr]_{+} ,
\qquad
A=\Lambda(\rho,H)\times
\begin{cases}
\dfrac12\sqrt{(D^{2}-1)\log\dfrac{D^{2}-1}{4p_{\mathrm{fail}}}} & \text{benchmark, Eq.~\eqref{eq:S-epsilonM}},\\[2ex]
\dfrac{\kappa}{\sqrt2}\sqrt{D}\,3^{n}\sqrt{\log\dfrac{2\cdot3^{n}}{p_{\mathrm{fail}}}} & \text{Pauli QST, Eq.~\eqref{eq:S-QST-bound}},
\end{cases}
\label{eq:S-A}
\end{equation}
to be maximised over $M\in(A^{2},N)$. In the region where the bracket is positive,
\begin{equation}
g'(M)=-\frac1N\Bigl(1-\frac{A}{\sqrt M}\Bigr)+\Bigl(1-\frac MN\Bigr)\frac{A}{2M^{3/2}} .
\label{eq:S-gprime}
\end{equation}
When $M\ll N$ and $A/\sqrt M\ll1$---both conditions are verified a posteriori---the stationarity condition $g'(M^{\ast})=0$ reduces to $1/N=A/(2M^{\ast\,3/2})$, i.e.
\begin{equation}
M^{\ast}\simeq\Bigl(\frac{AN}{2}\Bigr)^{2/3} ,
\label{eq:S-Mstar}
\end{equation}
which indeed satisfies $M^{\ast}/N\propto N^{-1/3}\to0$ and $A/\sqrt{M^{\ast}}\propto N^{-1/3}\to0$ at fixed $A$; the underlying requirement $N\gg A^{2}$ is the precise meaning of the scarce-resource regime. Substituting back,
\begin{equation}
\etaI^{\max}
\simeq1-\frac{M^{\ast}}{N}-\frac{A}{\sqrt{M^{\ast}}}
=1-\Bigl(\frac{A^{2}}{4}\Bigr)^{\!1/3}N^{-1/3}-\bigl(2A^{2}\bigr)^{1/3}N^{-1/3}
=1-\frac{3}{2^{2/3}}\,A^{2/3}\,N^{-1/3}+\mathcal{O}\bigl(N^{-2/3}\bigr) ,
\label{eq:S-etamax}
\end{equation}
where the two terms contribute $2^{-2/3}A^{2/3}N^{-1/3}$ and $2^{1/3}A^{2/3}N^{-1/3}$ respectively, and $2^{-2/3}+2^{1/3}=3\cdot2^{-2/3}$. This is Eq.~(15) of the main text. The $N^{-1/3}$ approach to the oracle limit is slower than either ingredient alone---the standard quantum limit $M^{-1/2}$ of learning, and the linear cost $1-M/N$ of consuming copies---because the optimum must balance the two; it is, in this sense, the thermodynamic price of inference. For the Pauli-QST realisation the same exponents hold with the prefactor $A^{2/3}\propto2^{n/3}\,3^{2n/3}$, exponential in the number of qubits: this is the curse of dimensionality of generic tomography, and precisely the gap that task-adapted reconstruction is designed to close.

\section{E --- Sharpness, refinements, and scope}
\label{sec:limitations}

We list here all the subtle points at which the arguments above are approximate or rest on assumptions, together with the available refinements. (i) Lemma~\ref{lem:lipschitz} is linear in $\D(\rho,\rhoI)$ with the universal constant $2\Delta_{H}$; the constant cannot be improved in general (consider a two-level system with $\rho$ nearly passive and a reconstruction error aligned with the energy basis), but for states with well-separated spectra a sharper, gap-dependent bound holds, derived below. (ii) The truncation of $-\log\F$ at second order in Eq.~\eqref{eq:S-F-expansion} is legitimate for $\delta_{i}\to0$, i.e.\ deep in the regime $M\gg A^{2}$; the third-order coefficient produces relative corrections $\mathcal{O}(\delta_{i})$ to Eq.~\eqref{eq:S-MF-single}, and a fully non-perturbative multi-parameter version of the finite-sample theory of Ref.~\cite{Meyer2025} is not currently available. (iii) Hypothesis (H) of Sec.~\ref{sec:multi-MF}---joint per-direction attainability at finite $M$---has been discussed there at length, together with its worst-case repair (sequential batching, an extra $\sqrt{D^{2}-1}$) and its asymptotic justification (Holevo attainability, an extra $R(\rho)\le2$). (iv) Theorem~1 is a statement at confidence $1-p_{\mathrm{fail}}$; on the complementary event the inferred unitary can even raise the energy of the extraction copies, so $\WI$ may be negative. Since $\D\le1$ always, Lemma~\ref{lem:decomp} gives the deterministic worst case $\WI/\Werg\ge1-\Lambda$, and averaging over the two events,
\begin{equation}
\mathbb{E}\bigl[\etaI\bigr]\ge\Bigl(1-\frac MN\Bigr)\Bigl[1-\Lambda(\rho,H)\bigl(\varepsilon_{M}+p_{\mathrm{fail}}\bigr)\Bigr] ,
\label{eq:S-expectation}
\end{equation}
so that choosing $p_{\mathrm{fail}}\lesssim\varepsilon_{M}$ makes the confidence overhead subleading. (v) Assumption (A4) restricts the demon to product unitaries on the extraction copies; collective unitaries would replace the single-copy ergotropy by its total (asymptotically, bound-ergotropy-free) counterpart, a strictly larger benchmark, and our analysis then applies with $\Werg$ reinterpreted accordingly. (vi) All concentration statements use the i.i.d.\ assumption (A1); correlated sources require different tools, see e.g.~\cite{Mortimer2026}. (vii) The Hamiltonian is assumed known exactly (A3); an uncertainty $\delta H$ propagates additively into both $\Delta_{H}$ and $\Werg$ and hence into $\Lambda$, and could be handled by combining the present protocol with Hamiltonian learning.

We conclude this section with the refinement announced in point (i), valid under the additional hypothesis
\begin{quote}
(A6) the spectrum of $\rho$ is non-degenerate, with gaps $\gamma_{k}:=\min_{j\ne k}\abs{r_{k}-r_{j}}>0$.
\end{quote}
Write $E:=\rhoI-\rho$ and let $P_{k}=\ketbra{r_{k}}{r_{k}}$ and $\Prec{k}=\ketbra{\rrec{k}}{\rrec{k}}$ be the spectral projectors of $\rho$ and $\rhoI$. Assume $\norm{E}_{\infty}<\gamma_{k}/2$. By Weyl's inequality, $\abs{\rrec{j}-r_{j}}\le\norm{E}_{\infty}$ for every $j$, so the $k$-th eigenvalue of $\rhoI$ remains separated from the rest of the spectrum of $\rho$ by at least $\gamma_{k}-\norm{E}_{\infty}>\gamma_{k}/2$, and in particular the ordering of the eigenvalues is preserved. The Davis--Kahan $\sin\Theta$ theorem~\cite{DavisKahan1970,Bhatia1997} then bounds the angle $\theta_{k}$ between the two eigenvectors by $\sin\theta_{k}\le\norm{E}_{\infty}/(\gamma_{k}-\norm{E}_{\infty})$, and since for rank-one projectors $\norm{\Prec{k}-P_{k}}_{\infty}=\sin\theta_{k}$,
\begin{equation}
\bigl\|\Prec{k}-P_{k}\bigr\|_{\infty}\;\le\;\frac{\norm{E}_{\infty}}{\gamma_{k}-\norm{E}_{\infty}}\;\le\;\frac{2\norm{E}_{\infty}}{\gamma_{k}} .
\label{eq:S-DK}
\end{equation}
(The weakened denominator $\gamma_{k}/2$ is not cosmetic: the na\"ive bound $\norm{E}_{\infty}/\gamma_{k}$, without the factor $2$, is violated already for qubit examples with $\norm{E}_{\infty}$ close to $\gamma_{k}/2$.) To convert Eq.~\eqref{eq:S-DK} into a deficit bound, express the diagonal matrix elements of $\sigma_{\star}$ and $\sigmaI $ in the energy eigenbasis. From Eqs.~\eqref{eq:S-sigmastar} and \eqref{eq:S-sigmaI}, using $\UWIdag \ket{\epsilon_{j}}=\ket{\rrec{j}}$,
\begin{equation}
(\sigma_{\star})_{jj}=r_{j}=\Tr[\rho P_{j}] ,
\qquad
(\sigmaI )_{jj}=\bra{\rrec{j}}\rho\ket{\rrec{j}}=\Tr[\rho \Prec{j}] ,
\label{eq:S-diagonals}
\end{equation}
hence, by H\"older's inequality with $\norm{\rho}_{1}=1$ and Eq.~\eqref{eq:S-DK},
\begin{equation}
\abs{(\sigmaI )_{jj}-(\sigma_{\star})_{jj}}
=\abs{\Tr\bigl[\rho\,(\Prec{j}-P_{j})\bigr]}
\le\bigl\|\Prec{j}-P_{j}\bigr\|_{\infty}
\le\frac{2\norm{E}_{\infty}}{\gamma_{j}} .
\label{eq:S-diag-bound}
\end{equation}
Now $\Delta W=\Tr[H\sigmaI ]-\Tr[H\sigma_{\star}]=\sum_{j}\epsilon_{j}\bigl[(\sigmaI )_{jj}-(\sigma_{\star})_{jj}\bigr]$, and since both $\sigmaI $ and $\sigma_{\star}$ have unit trace the differences sum to zero, so any constant may be subtracted from the $\epsilon_{j}$; choosing $\bar\epsilon:=D^{-1}\sum_{j}\epsilon_{j}$ and applying the triangle inequality with Eq.~\eqref{eq:S-diag-bound},
\begin{equation}
\Delta W\;\le\;2\,\norm{\rhoI-\rho}_{\infty}\,\sum_{j=1}^{D}\frac{\abs{\epsilon_{j}-\bar\epsilon}}{\gamma_{j}} .
\label{eq:S-DK-deficit}
\end{equation}
Because $\norm{E}_{\infty}\le\norm{E}_{1}=2\D(\rho,\rhoI)$, and can be far smaller for high-rank perturbations, Eq.~\eqref{eq:S-DK-deficit} improves on Lemma~\ref{lem:lipschitz} whenever the spectrum of $\rho$ is well-gapped, $\sum_{j}\abs{\epsilon_{j}-\bar\epsilon}/\gamma_{j}\lesssim\Delta_{H}/\norm{E}_{\infty}\cdot\D$; in the opposite, near-degenerate limit it diverges and Lemma~\ref{lem:decomp}---which requires no spectral assumption whatsoever---takes over. Interpolating between the two regimes would require operator-angle techniques in the spirit of Ref.~\cite{Bhatia1997} and is left for future work.

\section{F --- Regime of validity}
\label{sec:regime}

The lower envelope of Theorem~1 is informative if and only if $\Lambda(\rho,H)\,\varepsilon_{M}<1$, i.e.
\begin{equation}
\varepsilon_{M}<\frac{1}{\Lambda(\rho,H)}=\frac{\Werg(\rho,H)}{2\Delta_{H}}\;\le\;\frac12 ,
\label{eq:S-threshold-eps}
\end{equation}
the last inequality because $\Werg\le\Delta_{H}$ for any state (the mean energy cannot exceed $\epsilon_{D}$ nor the passive energy drop below $\epsilon_{1}$). This makes precise the reading of $1/\Lambda$ as a state-dependent threshold error, beyond which partial information certifies nothing. Inserting the benchmark~\eqref{eq:S-epsilonM}, the demon must invest at least
\begin{equation}
M\;>\;M_{\mathrm{thr}}:=\frac{\Lambda(\rho,H)^{2}}{4}\,(D^{2}-1)\log\frac{D^{2}-1}{4p_{\mathrm{fail}}}
=\frac{\Delta_{H}^{2}}{\Werg(\rho,H)^{2}}\,(D^{2}-1)\log\frac{D^{2}-1}{4p_{\mathrm{fail}}}
\label{eq:S-Mthr}
\end{equation}
copies in the reconstruction before any nontrivial guarantee on the extracted work becomes possible. For nearly passive states, $\Werg\to0$ and $M_{\mathrm{thr}}$ diverges: no finite information budget certifies work extraction from a state indistinguishable from passive, in agreement with the sample-complexity obstructions to black-box work extraction of Ref.~\cite{Chakraborty2025}. For maximally mixed or nearly flat spectra the tomographic model degenerates ($\mathcal{F}_{Q}$ becomes singular along some directions and the expansion~\eqref{eq:S-F-expansion} loses accuracy), and Eq.~\eqref{eq:S-Mthr} should be read as necessary rather than sufficient, consistently with the known difficulty of reconstructing highly mixed states~\cite{HHJLW2017}. For $M\gg M_{\mathrm{thr}}$, finally, the efficiency follows the asymptotic law~\eqref{eq:S-etamax}.

\setcounter{NAT@ctr}{0}
\begin{supplementbib}{99}
\footnotesize

\bibitem{Meyer2025} J.~J.~Meyer, S.~Khatri, D.~Stilck Fran\c{c}a, J.~Eisert, and P.~Faist, \emph{Quantum metrology in the finite-sample regime}, PRX Quantum \textbf{6}, 030336 (2025).

\bibitem{Helstrom1976} C.~W.~Helstrom, \emph{Quantum Detection and Estimation Theory} (Academic Press, New York, 1976).

\bibitem{Holevo1976} A.~S.~Holevo, \emph{Probabilistic and Statistical Aspects of Quantum Theory} (North-Holland, Amsterdam, 1982).

\bibitem{YangChiribellaHayashi2019} Y.~Yang, G.~Chiribella, and M.~Hayashi, \emph{Attaining the ultimate precision limit in quantum state estimation}, Commun.\ Math.\ Phys.\ \textbf{368}, 223 (2019).

\bibitem{AlbarelliFrielDatta2019} F.~Albarelli, J.~F.~Friel, and A.~Datta, \emph{Evaluating the Holevo Cram\'er--Rao bound for multiparameter quantum metrology}, Phys.\ Rev.\ Lett.\ \textbf{123}, 200503 (2019).

\bibitem{BelliardoGiovannetti2021} F.~Belliardo and V.~Giovannetti, \emph{Incompatibility in quantum parameter estimation}, New J.\ Phys.\ \textbf{23}, 063055 (2021).

\bibitem{BraunsteinCaves1994} S.~L.~Braunstein and C.~M.~Caves, \emph{Statistical distance and the geometry of quantum states}, Phys.\ Rev.\ Lett.\ \textbf{72}, 3439 (1994).

\bibitem{Hubner1992} M.~H\"ubner, \emph{Explicit computation of the Bures distance for density matrices}, Phys.\ Lett.\ A \textbf{163}, 239 (1992).

\bibitem{Pagliaro2026} E.~Pagliaro, L.~Zambrano, M.~Alimuddin, A.~Hamma, A.~Ac\'\i n, and D.~Farina, \emph{Certifying ergotropy under partial information}, arXiv:2603.18828 (2026).

\bibitem{Zambrano2026} L.~Zambrano, T.~Parella-Dilm\'e, A.~Ac\'\i n, and D.~Farina, \emph{Certification of quantum properties with imperfect measurements}, arXiv:2601.16570 (2026).

\bibitem{Mortimer2026} L.~Mortimer, L.~Zambrano, A.~Ac\'\i n, and D.~Farina, \emph{Bounding many-body properties under partial information and finite measurement statistics}, arXiv:2601.10408 (2026).

\bibitem{Fuchs1999} C.~A.~Fuchs and J.~van de Graaf, \emph{Cryptographic distinguishability measures for quantum-mechanical states}, IEEE Trans.\ Inf.\ Theory \textbf{45}, 1216 (1999).

\bibitem{Fannes1973} M.~Fannes, \emph{A continuity property of the entropy density for spin lattice systems}, Commun.\ Math.\ Phys.\ \textbf{31}, 291 (1973).

\bibitem{Audenaert2007} K.~M.~R.~Audenaert, \emph{A sharp continuity estimate for the von Neumann entropy}, J.\ Phys.\ A: Math.\ Theor.\ \textbf{40}, 8127 (2007).

\bibitem{Hoeffding1963} W.~Hoeffding, \emph{Probability inequalities for sums of bounded random variables}, J.\ Am.\ Stat.\ Assoc.\ \textbf{58}, 13 (1963).

\bibitem{BHC} J.~Bretagnolle and C.~Huber-Carol, \emph{Estimation des densit\'es: risque minimax}, Z.\ Wahrscheinlichkeitstheorie Verw.\ Gebiete \textbf{47}, 119 (1979).

\bibitem{HHJLW2017} J.~Haah, A.~W.~Harrow, Z.~Ji, X.~Liu, and N.~Wu, \emph{Sample-optimal tomography of quantum states}, IEEE Trans.\ Inf.\ Theory \textbf{63}, 5628 (2017).

\bibitem{DavisKahan1970} C.~Davis and W.~M.~Kahan, \emph{The rotation of eigenvectors by a perturbation. III}, SIAM J.\ Numer.\ Anal.\ \textbf{7}, 1 (1970).

\bibitem{Bhatia1997} R.~Bhatia, \emph{Matrix Analysis} (Springer, New York, 1997).

\bibitem{Chakraborty2025} S.~Chakraborty, S.~Das, A.~Ghorui, S.~Hazra, and U.~Singh, \emph{Sample complexity of black box work extraction}, Quantum Sci.\ Technol.\ \textbf{10}, 045070 (2025).

\end{supplementbib}

%% file: refs.bib
@book{maxwell2001theory,
  author    = {Maxwell, James Clerk and Pesic, Peter},
  title     = {Theory of Heat},
  publisher = {Dover Publications (Courier Corporation)},
  address   = {Mineola, NY},
  year      = {2001},
  note      = {Reprint of the 1888 edition},
  url       = {https://store.doverpublications.com/products/9780486417356}
}

@article{szilard1929entropieverminderung,
  author  = {Szilard, Leo},
  title   = {{\"U}ber die Entropieverminderung in einem thermodynamischen System bei Eingriffen intelligenter Wesen},
  journal = {Zeitschrift f{\"u}r Physik},
  volume  = {53},
  pages   = {840--856},
  year    = {1929},
  doi     = {10.1007/BF01341281},
  url     = {https://link.springer.com/article/10.1007/BF01341281}
}

@article{szilard1964decrease,
  author  = {Szilard, Leo},
  title   = {On the decrease of entropy in a thermodynamic system by the intervention of intelligent beings},
  journal = {Behavioral Science},
  volume  = {9},
  number  = {4},
  pages   = {301--310},
  year    = {1964},
  doi     = {10.1002/bs.3830090402},
  url     = {https://onlinelibrary.wiley.com/doi/10.1002/bs.3830090402}
}

@article{landauer1961irreversibility,
  author  = {Landauer, Rolf},
  title   = {Irreversibility and heat generation in the computing process},
  journal = {IBM Journal of Research and Development},
  volume  = {5},
  number  = {3},
  pages   = {183--191},
  year    = {1961},
  doi     = {10.1147/rd.53.0183},
  url     = {https://ieeexplore.ieee.org/document/5392446}
}

@article{bennett1973logical,
  author  = {Bennett, Charles H.},
  title   = {Logical reversibility of computation},
  journal = {IBM Journal of Research and Development},
  volume  = {17},
  number  = {6},
  pages   = {525--532},
  year    = {1973},
  doi     = {10.1147/rd.176.0525},
  url     = {https://ieeexplore.ieee.org/document/5391327}
}

@article{bennett1982thermodynamics,
  author  = {Bennett, Charles H.},
  title   = {The thermodynamics of computation---a review},
  journal = {International Journal of Theoretical Physics},
  volume  = {21},
  number  = {12},
  pages   = {905--940},
  year    = {1982},
  doi     = {10.1007/BF02084158},
  url     = {https://link.springer.com/article/10.1007/BF02084158}
}

@article{allahverdyan2004maximal,
  author  = {Allahverdyan, Armen E. and Balian, Roger and Nieuwenhuizen, Theo M.},
  title   = {Maximal work extraction from finite quantum systems},
  journal = {Europhysics Letters},
  volume  = {67},
  number  = {4},
  pages   = {565--571},
  year    = {2004},
  doi     = {10.1209/epl/i2004-10101-2},
  url     = {https://iopscience.iop.org/article/10.1209/epl/i2004-10101-2}
}

@article{meyer2025quantum,
  author  = {Meyer, Johannes Jakob and Khatri, Sumeet and Stilck Fran{\c c}a, Daniel and Eisert, Jens and Faist, Philippe},
  title   = {Quantum metrology in the finite-sample regime},
  journal = {PRX Quantum},
  volume  = {6},
  pages   = {030336},
  year    = {2025},
  doi     = {10.1103/qbn1-p6bq},
  note    = {arXiv:2307.06370},
  url     = {https://journals.aps.org/prxquantum/abstract/10.1103/qbn1-p6bq}
}

@article{pagliaro2026certifying,
  author  = {Pagliaro, Elisa and Zambrano, Leonardo and Alimuddin, Mir and Hamma, Alioscia and Ac{\'\i}n, Antonio and Farina, Donato},
  title   = {Certifying ergotropy under partial information},
  journal = {arXiv preprint},
  year    = {2026},
  note    = {arXiv:2603.18828},
  url     = {https://arxiv.org/abs/2603.18828}
}

@article{zambrano2026certification,
  author  = {Zambrano, Leonardo and Parella-Dilm{\'e}, Teresa and Ac{\'\i}n, Antonio and Farina, Donato},
  title   = {Certification of quantum properties with imperfect measurements},
  journal = {Quantum Science and Technology},
  year    = {2026},
  doi     = {10.1088/2058-9565/ae7753},
  note    = {arXiv:2601.16570},
  url     = {https://iopscience.iop.org/article/10.1088/2058-9565/ae7753}
}

@article{mortimer2026bounding,
  author  = {Mortimer, Luke and Zambrano, Leonardo and Ac{\'\i}n, Antonio and Farina, Donato},
  title   = {Bounding many-body properties under partial information and finite measurement statistics},
  journal = {arXiv preprint},
  year    = {2026},
  note    = {arXiv:2601.10408},
  url     = {https://arxiv.org/abs/2601.10408}
}

@article{chakraborty2025sample,
  author  = {Chakraborty, Shantanav and Das, Siddhartha and Ghorui, Arnab and Hazra, Soumyabrata and Singh, Uttam},
  title   = {Sample complexity of black box work extraction},
  journal = {Quantum Science and Technology},
  volume  = {10},
  pages   = {045070},
  year    = {2025},
  doi     = {10.1088/2058-9565/ae0e4d},
  note    = {arXiv:2412.02673},
  url     = {https://iopscience.iop.org/article/10.1088/2058-9565/ae0e4d}
}

@article{binosi2024tailor,
  author  = {Binosi, Daniele and Garberoglio, Giovanni and Maragnano, Diego and Dapor, Maurizio and Liscidini, Marco},
  title   = {A tailor-made quantum state tomography approach},
  journal = {APL Quantum},
  volume  = {1},
  number  = {3},
  pages   = {036112},
  year    = {2024},
  doi     = {10.1063/5.0219143},
  note    = {arXiv:2401.12864},
  url     = {https://pubs.aip.org/aip/apq/article/1/3/036112/3304804}
}

@article{caruccio2025experimental,
  author  = {Caruccio, Eugenio and Maragnano, Diego and Rodari, Giovanni and Picus, Davide and Garberoglio, Giovanni and Binosi, Daniele and Albiero, Riccardo and Di Giano, Niki and Ceccarelli, Francesco and Corrielli, Giacomo and Spagnolo, Nicol{\`o} and Osellame, Roberto and Dapor, Maurizio and Liscidini, Marco and Sciarrino, Fabio},
  title   = {Experimental verification of threshold quantum state tomography on a fully-reconfigurable photonic integrated circuit},
  journal = {npj Quantum Information},
  volume  = {11},
  pages   = {173},
  year    = {2025},
  doi     = {10.1038/s41534-025-01111-z},
  note    = {arXiv:2504.05079},
  url     = {https://www.nature.com/articles/s41534-025-01111-z}
}

@article{james2001qst,
  author  = {James, Daniel F. V. and Kwiat, Paul G. and Munro, William J. and White, Andrew G.},
  title   = {Measurement of qubits},
  journal = {Physical Review A},
  volume  = {64},
  number  = {5},
  pages   = {052312},
  year    = {2001},
  doi     = {10.1103/PhysRevA.64.052312},
  url     = {https://link.aps.org/doi/10.1103/PhysRevA.64.052312}
}

@article{fuchs1999cryptographic,
  author  = {Fuchs, Christopher A. and van de Graaf, Jeroen},
  title   = {Cryptographic distinguishability measures for quantum-mechanical states},
  journal = {IEEE Transactions on Information Theory},
  volume  = {45},
  number  = {4},
  pages   = {1216--1227},
  year    = {1999},
  doi     = {10.1109/18.761271},
  url     = {https://ieeexplore.ieee.org/document/761271}
}

@article{sagawa2008second,
  author  = {Sagawa, Takahiro and Ueda, Masahito},
  title   = {Second law of thermodynamics with discrete quantum feedback control},
  journal = {Physical Review Letters},
  volume  = {100},
  pages   = {080403},
  year    = {2008},
  doi     = {10.1103/PhysRevLett.100.080403},
  url     = {https://link.aps.org/doi/10.1103/PhysRevLett.100.080403}
}

@article{parrondo2015thermodynamics,
  author  = {Parrondo, Juan M. R. and Horowitz, Jordan M. and Sagawa, Takahiro},
  title   = {Thermodynamics of information},
  journal = {Nature Physics},
  volume  = {11},
  number  = {2},
  pages   = {131--139},
  year    = {2015},
  doi     = {10.1038/nphys3230},
  url     = {https://www.nature.com/articles/nphys3230}
}

@article{maruyama2009colloquium,
  author  = {Maruyama, Koji and Nori, Franco and Vedral, Vlatko},
  title   = {Colloquium: The physics of Maxwell's demon and information},
  journal = {Reviews of Modern Physics},
  volume  = {81},
  number  = {1},
  pages   = {1--23},
  year    = {2009},
  doi     = {10.1103/RevModPhys.81.1},
  url     = {https://link.aps.org/doi/10.1103/RevModPhys.81.1}
}

@article{toyabe2010experimental,
  author  = {Toyabe, Shoichi and Sagawa, Takahiro and Ueda, Masahito and Muneyuki, Eiro and Sano, Masaki},
  title   = {Experimental demonstration of information-to-energy conversion and validation of the generalized Jarzynski equality},
  journal = {Nature Physics},
  volume  = {6},
  number  = {12},
  pages   = {988--992},
  year    = {2010},
  doi     = {10.1038/nphys1821},
  url     = {https://www.nature.com/articles/nphys1821}
}

@article{berut2012experimental,
  author  = {B{\'e}rut, Antoine and Arakelyan, Artak and Petrosyan, Artyom and Ciliberto, Sergio and Dillenschneider, Raoul and Lutz, Eric},
  title   = {Experimental verification of Landauer's principle linking information and thermodynamics},
  journal = {Nature},
  volume  = {483},
  number  = {7388},
  pages   = {187--189},
  year    = {2012},
  doi     = {10.1038/nature10872},
  url     = {https://www.nature.com/articles/nature10872}
}

@article{koski2014experimental,
  author  = {Koski, Jonne V. and Maisi, Ville F. and Pekola, Jukka P. and Averin, Dmitri V.},
  title   = {Experimental realization of a Szilard engine with a single electron},
  journal = {Proceedings of the National Academy of Sciences},
  volume  = {111},
  number  = {38},
  pages   = {13786--13789},
  year    = {2014},
  doi     = {10.1073/pnas.1406966111},
  url     = {https://www.pnas.org/doi/10.1073/pnas.1406966111}
}

@article{delrio2011thermodynamic,
  author  = {del Rio, L{\'\i}dia and {\AA}berg, Johan and Renner, Renato and Dahlsten, Oscar and Vedral, Vlatko},
  title   = {The thermodynamic meaning of negative entropy},
  journal = {Nature},
  volume  = {474},
  number  = {7349},
  pages   = {61--63},
  year    = {2011},
  doi     = {10.1038/nature10395},
  url     = {https://www.nature.com/articles/nature10395}
}

@article{goold2016role,
  author  = {Goold, John and Huber, Marcus and Riera, Arnau and del Rio, L{\'\i}dia and Skrzypczyk, Paul},
  title   = {The role of quantum information in thermodynamics---a topical review},
  journal = {Journal of Physics A: Mathematical and Theoretical},
  volume  = {49},
  number  = {14},
  pages   = {143001},
  year    = {2016},
  doi     = {10.1088/1751-8113/49/14/143001},
  url     = {https://iopscience.iop.org/article/10.1088/1751-8113/49/14/143001}
}

@article{vinjanampathy2016quantum,
  author  = {Vinjanampathy, Sai and Anders, Janet},
  title   = {Quantum thermodynamics},
  journal = {Contemporary Physics},
  volume  = {57},
  number  = {4},
  pages   = {545--579},
  year    = {2016},
  doi     = {10.1080/00107514.2016.1201896},
  url     = {https://www.tandfonline.com/doi/full/10.1080/00107514.2016.1201896}
}

@article{pusz1978passive,
  author  = {Pusz, Wojciech and Woronowicz, Stanis{\l}aw L.},
  title   = {Passive states and KMS states for general quantum systems},
  journal = {Communications in Mathematical Physics},
  volume  = {58},
  number  = {3},
  pages   = {273--290},
  year    = {1978},
  doi     = {10.1007/BF01614224},
  url     = {https://link.springer.com/article/10.1007/BF01614224}
}

@article{lenard1978thermodynamical,
  author  = {Lenard, Andrew},
  title   = {Thermodynamical proof of the Gibbs formula for elementary quantum systems},
  journal = {Journal of Statistical Physics},
  volume  = {19},
  number  = {6},
  pages   = {575--586},
  year    = {1978},
  doi     = {10.1007/BF01011769},
  url     = {https://link.springer.com/article/10.1007/BF01011769}
}

@article{skrzypczyk2014work,
  author  = {Skrzypczyk, Paul and Short, Anthony J. and Popescu, Sandu},
  title   = {Work extraction and thermodynamics for individual quantum systems},
  journal = {Nature Communications},
  volume  = {5},
  pages   = {4185},
  year    = {2014},
  doi     = {10.1038/ncomms5185},
  url     = {https://www.nature.com/articles/ncomms5185}
}

@article{francica2020quantum,
  author  = {Francica, Gianluca and Binder, Felix C. and Guarnieri, Giacomo and Mitchison, Mark T. and Goold, John and Plastina, Francesco},
  title   = {Quantum coherence and ergotropy},
  journal = {Physical Review Letters},
  volume  = {125},
  pages   = {180603},
  year    = {2020},
  doi     = {10.1103/PhysRevLett.125.180603},
  url     = {https://link.aps.org/doi/10.1103/PhysRevLett.125.180603}
}

@article{alicki2013entanglement,
  author  = {Alicki, Robert and Fannes, Mark},
  title   = {Entanglement boost for extractable work from ensembles of quantum batteries},
  journal = {Physical Review E},
  volume  = {87},
  pages   = {042123},
  year    = {2013},
  doi     = {10.1103/PhysRevE.87.042123},
  url     = {https://link.aps.org/doi/10.1103/PhysRevE.87.042123}
}

@article{andolina2019extractable,
  author  = {Andolina, Gian Marcello and Keck, Maximilian and Mari, Andrea and Giovannetti, Vittorio and Polini, Marco},
  title   = {Extractable work, the role of correlations, and asymptotic freedom in quantum batteries},
  journal = {Physical Review Letters},
  volume  = {122},
  pages   = {047702},
  year    = {2019},
  doi     = {10.1103/PhysRevLett.122.047702},
  url     = {https://link.aps.org/doi/10.1103/PhysRevLett.122.047702}
}

@article{campaioli2024colloquium,
  author  = {Campaioli, Francesco and Gherardini, Stefano and Quach, James Q. and Polini, Marco and Andolina, Gian Marcello},
  title   = {Colloquium: Quantum batteries},
  journal = {Reviews of Modern Physics},
  volume  = {96},
  number  = {3},
  pages   = {031001},
  year    = {2024},
  doi     = {10.1103/RevModPhys.96.031001},
  url     = {https://link.aps.org/doi/10.1103/RevModPhys.96.031001}
}

@article{horodecki2013fundamental,
  author  = {Horodecki, Micha{\l} and Oppenheim, Jonathan},
  title   = {Fundamental limitations for quantum and nanoscale thermodynamics},
  journal = {Nature Communications},
  volume  = {4},
  pages   = {2059},
  year    = {2013},
  doi     = {10.1038/ncomms3059},
  url     = {https://www.nature.com/articles/ncomms3059}
}

@article{brandao2015second,
  author  = {Brand{\~a}o, Fernando and Horodecki, Micha{\l} and Ng, Nelly and Oppenheim, Jonathan and Wehner, Stephanie},
  title   = {The second laws of quantum thermodynamics},
  journal = {Proceedings of the National Academy of Sciences},
  volume  = {112},
  number  = {11},
  pages   = {3275--3279},
  year    = {2015},
  doi     = {10.1073/pnas.1411728112},
  url     = {https://www.pnas.org/doi/10.1073/pnas.1411728112}
}

@article{baumgratz2014quantifying,
  author  = {Baumgratz, Tillmann and Cramer, Marcus and Plenio, Martin B.},
  title   = {Quantifying coherence},
  journal = {Physical Review Letters},
  volume  = {113},
  pages   = {140401},
  year    = {2014},
  doi     = {10.1103/PhysRevLett.113.140401},
  url     = {https://link.aps.org/doi/10.1103/PhysRevLett.113.140401}
}

@article{faist2015minimal,
  author  = {Faist, Philippe and Dupuis, Fr{\'e}d{\'e}ric and Oppenheim, Jonathan and Renner, Renato},
  title   = {The minimal work cost of information processing},
  journal = {Nature Communications},
  volume  = {6},
  pages   = {7669},
  year    = {2015},
  doi     = {10.1038/ncomms8669},
  url     = {https://www.nature.com/articles/ncomms8669}
}

@article{gross2010quantum,
  author  = {Gross, David and Liu, Yi-Kai and Flammia, Steven T. and Becker, Stephen and Eisert, Jens},
  title   = {Quantum state tomography via compressed sensing},
  journal = {Physical Review Letters},
  volume  = {105},
  pages   = {150401},
  year    = {2010},
  doi     = {10.1103/PhysRevLett.105.150401},
  url     = {https://link.aps.org/doi/10.1103/PhysRevLett.105.150401}
}

@article{huang2020predicting,
  author  = {Huang, Hsin-Yuan and Kueng, Richard and Preskill, John},
  title   = {Predicting many properties of a quantum system from very few measurements},
  journal = {Nature Physics},
  volume  = {16},
  number  = {10},
  pages   = {1050--1057},
  year    = {2020},
  doi     = {10.1038/s41567-020-0932-7},
  url     = {https://www.nature.com/articles/s41567-020-0932-7}
}

@article{aaronson2020shadow,
  author  = {Aaronson, Scott},
  title   = {Shadow tomography of quantum states},
  journal = {SIAM Journal on Computing},
  volume  = {49},
  number  = {5},
  pages   = {STOC18-368--STOC18-394},
  year    = {2020},
  doi     = {10.1137/18M120275X},
  note    = {Preliminary version in Proc. 50th ACM STOC (2018), 325--338},
  url     = {https://epubs.siam.org/doi/10.1137/18M120275X}
}

@article{elben2023randomized,
  author  = {Elben, Andreas and Flammia, Steven T. and Huang, Hsin-Yuan and Kueng, Richard and Preskill, John and Vermersch, Beno{\^\i}t and Zoller, Peter},
  title   = {The randomized measurement toolbox},
  journal = {Nature Reviews Physics},
  volume  = {5},
  number  = {1},
  pages   = {9--24},
  year    = {2023},
  doi     = {10.1038/s42254-022-00535-2},
  url     = {https://www.nature.com/articles/s42254-022-00535-2}
}

@article{giovannetti2011advances,
  author  = {Giovannetti, Vittorio and Lloyd, Seth and Maccone, Lorenzo},
  title   = {Advances in quantum metrology},
  journal = {Nature Photonics},
  volume  = {5},
  number  = {4},
  pages   = {222--229},
  year    = {2011},
  doi     = {10.1038/nphoton.2011.35},
  url     = {https://www.nature.com/articles/nphoton.2011.35}
}

@article{paris2009quantum,
  author  = {Paris, Matteo G. A.},
  title   = {Quantum estimation for quantum technology},
  journal = {International Journal of Quantum Information},
  volume  = {7},
  number  = {supp01},
  pages   = {125--137},
  year    = {2009},
  doi     = {10.1142/S0219749909004839},
  url     = {https://www.worldscientific.com/doi/10.1142/S0219749909004839}
}

@article{szczykulska2016multi,
  author  = {Szczykulska, Magdalena and Baumgratz, Tillmann and Datta, Animesh},
  title   = {Multi-parameter quantum metrology},
  journal = {Advances in Physics: X},
  volume  = {1},
  number  = {4},
  pages   = {621--639},
  year    = {2016},
  doi     = {10.1080/23746149.2016.1230476},
  url     = {https://www.tandfonline.com/doi/full/10.1080/23746149.2016.1230476}
}

@misc{SM,
  note = {See Supplemental Material for the full proof of the ergotropic-deficit bound, the derivation of the closed-form reconstruction precision from the finite-sample bounds of Ref.~\cite{meyer2025quantum} together with its critical comparison against the Holevo Cram\'er--Rao theory of multi-parameter estimation, the Hoeffding-type certificates for laboratory tomographic protocols, the derivation of the optimal copy allocation, a gap-dependent refinement of the deficit bound, and the quantitative regime of validity.}
}


%% file: refs_supplement.bib
@book{Helstrom1976,
  author    = {Helstrom, Carl W.},
  title     = {Quantum Detection and Estimation Theory},
  publisher = {Academic Press},
  address   = {New York},
  year      = {1976},
  url       = {https://www.sciencedirect.com/bookseries/mathematics-in-science-and-engineering/vol/123}
}

@book{Holevo1976,
  author    = {Holevo, Alexander S.},
  title     = {Probabilistic and Statistical Aspects of Quantum Theory},
  publisher = {North-Holland},
  address   = {Amsterdam},
  year      = {1982},
  note      = {2nd English ed., Edizioni della Normale, Pisa (2011), DOI 10.1007/978-88-7642-378-9},
  url       = {https://link.springer.com/book/10.1007/978-88-7642-378-9}
}

@article{YangChiribellaHayashi2019,
  author  = {Yang, Yuxiang and Chiribella, Giulio and Hayashi, Masahito},
  title   = {Attaining the ultimate precision limit in quantum state estimation},
  journal = {Communications in Mathematical Physics},
  volume  = {368},
  number  = {1},
  pages   = {223--293},
  year    = {2019},
  doi     = {10.1007/s00220-019-03433-4},
  url     = {https://link.springer.com/article/10.1007/s00220-019-03433-4}
}

@article{AlbarelliFrielDatta2019,
  author  = {Albarelli, Francesco and Friel, Jamie F. and Datta, Animesh},
  title   = {Evaluating the Holevo Cram\'er--Rao bound for multiparameter quantum metrology},
  journal = {Physical Review Letters},
  volume  = {123},
  pages   = {200503},
  year    = {2019},
  doi     = {10.1103/PhysRevLett.123.200503},
  url     = {https://link.aps.org/doi/10.1103/PhysRevLett.123.200503}
}

@article{BelliardoGiovannetti2021,
  author  = {Belliardo, Federico and Giovannetti, Vittorio},
  title   = {Incompatibility in quantum parameter estimation},
  journal = {New Journal of Physics},
  volume  = {23},
  pages   = {063055},
  year    = {2021},
  doi     = {10.1088/1367-2630/ac04ca},
  url     = {https://iopscience.iop.org/article/10.1088/1367-2630/ac04ca}
}

@article{BraunsteinCaves1994,
  author  = {Braunstein, Samuel L. and Caves, Carlton M.},
  title   = {Statistical distance and the geometry of quantum states},
  journal = {Physical Review Letters},
  volume  = {72},
  pages   = {3439--3443},
  year    = {1994},
  doi     = {10.1103/PhysRevLett.72.3439},
  url     = {https://link.aps.org/doi/10.1103/PhysRevLett.72.3439}
}

@article{Hubner1992,
  author  = {H{\"u}bner, Matthias},
  title   = {Explicit computation of the Bures distance for density matrices},
  journal = {Physics Letters A},
  volume  = {163},
  number  = {4},
  pages   = {239--242},
  year    = {1992},
  doi     = {10.1016/0375-9601(92)91004-B},
  url     = {https://www.sciencedirect.com/science/article/pii/037596019291004B}
}

@article{Fannes1973,
  author  = {Fannes, Mark},
  title   = {A continuity property of the entropy density for spin lattice systems},
  journal = {Communications in Mathematical Physics},
  volume  = {31},
  number  = {4},
  pages   = {291--294},
  year    = {1973},
  doi     = {10.1007/BF01646490},
  url     = {https://link.springer.com/article/10.1007/BF01646490}
}

@article{Audenaert2007,
  author  = {Audenaert, Koenraad M. R.},
  title   = {A sharp continuity estimate for the von Neumann entropy},
  journal = {Journal of Physics A: Mathematical and Theoretical},
  volume  = {40},
  number  = {28},
  pages   = {8127--8136},
  year    = {2007},
  doi     = {10.1088/1751-8113/40/28/S18},
  url     = {https://iopscience.iop.org/article/10.1088/1751-8113/40/28/S18}
}

@article{Hoeffding1963,
  author  = {Hoeffding, Wassily},
  title   = {Probability inequalities for sums of bounded random variables},
  journal = {Journal of the American Statistical Association},
  volume  = {58},
  number  = {301},
  pages   = {13--30},
  year    = {1963},
  doi     = {10.1080/01621459.1963.10500830},
  url     = {https://www.tandfonline.com/doi/abs/10.1080/01621459.1963.10500830}
}

@article{BHC,
  author  = {Bretagnolle, Jean and Huber-Carol, Catherine},
  title   = {Estimation des densit\'es: risque minimax},
  journal = {Zeitschrift f{\"u}r Wahrscheinlichkeitstheorie und Verwandte Gebiete},
  volume  = {47},
  number  = {2},
  pages   = {119--137},
  year    = {1979},
  doi     = {10.1007/BF00535278},
  url     = {https://link.springer.com/article/10.1007/BF00535278}
}

@article{HHJLW2017,
  author  = {Haah, Jeongwan and Harrow, Aram W. and Ji, Zhengfeng and Liu, Xiaodi and Wu, Nengkun},
  title   = {Sample-optimal tomography of quantum states},
  journal = {IEEE Transactions on Information Theory},
  volume  = {63},
  number  = {9},
  pages   = {5628--5641},
  year    = {2017},
  doi     = {10.1109/TIT.2017.2719044},
  url     = {https://ieeexplore.ieee.org/document/7956181}
}

@article{DavisKahan1970,
  author  = {Davis, Chandler and Kahan, William M.},
  title   = {The rotation of eigenvectors by a perturbation. III},
  journal = {SIAM Journal on Numerical Analysis},
  volume  = {7},
  number  = {1},
  pages   = {1--46},
  year    = {1970},
  doi     = {10.1137/0707001},
  url     = {https://epubs.siam.org/doi/10.1137/0707001}
}

@book{Bhatia1997,
  author    = {Bhatia, Rajendra},
  title     = {Matrix Analysis},
  series    = {Graduate Texts in Mathematics},
  volume    = {169},
  publisher = {Springer},
  address   = {New York},
  year      = {1997},
  doi       = {10.1007/978-1-4612-0653-8},
  url       = {https://link.springer.com/book/10.1007/978-1-4612-0653-8}
}
